\documentclass[12pt]{article}
\usepackage{epsfig, graphics, subfigure}
\usepackage{amsfonts}
\usepackage{dsfont}
\usepackage{tikz}
\usepackage{amsthm}     
\usepackage{amsmath}
\usepackage{lscape}
\usepackage{natbib}

\usepackage[english]{babel} 
\usepackage[latin1]{inputenc}
\usepackage[T1]{fontenc} 

\usepackage{float}
\usepackage[a4paper, margin=1in]{geometry}
\newtheorem{prop}{Proposition}  
\newtheorem{Cor}{Corollary}

\newcommand{\Ben}{\begin{enumerate}}
\newcommand{\Een}{\end{enumerate}}
\newcommand{\Bit}{\begin{itemize}}
\newcommand{\Eit}{\end{itemize}}
\newcommand{\Beq}{\begin{equation}}
\newcommand{\Eeq}{\end{equation}}

\newcommand{\Ba}{\begin{align*}}
\newcommand{\Ea}{\end{align*}}
\newcommand{\Mb}{\mathbf}
\newtheorem{Th}{Theorem}
\newtheorem{Lem}{Lemma}
\newtheorem{Rq}{Remark}

\def \Frac {\displaystyle \frac }
\def \Int {\displaystyle \int }
\def \Sum {\displaystyle \sum }

\title{Diversification and Endogenous Financial Networks}
\author{Jean-Cyprien H\'eam \footnote{Autorit\'e de Contr\^ole Prudentiel et de R\'esolution (ACPR) and CREST. jean-cyprien.heam@acpr.banque-france.fr 
  } ~~~~ Erwan Koch\footnote{ISFA, CREST and ETH Zurich (Department of Mathematics, RiskLab). erwan.koch@math.ethz.ch 
   }}

\date{\today}

\begin{document}

\maketitle
\vspace{2cm}
\begin{abstract}

We test the hypothesis that interconnections across financial institutions can be explained by a diversification motive.
This idea stems from the empirical evidence of the existence of long-term exposures that cannot be explained by a liquidity motive (maturity or currency mismatch). We model endogenous interconnections of heterogenous financial institutions facing regulatory constraints using a maximization of their expected utility. Both theoretical and simulation-based results are compared to a stylized genuine financial network. The diversification motive appears to plausibly explain interconnections among key players. Using our model, the impact of regulation on interconnections between banks -currently discussed at the Basel Committee on Banking Supervision- is analyzed.

\medskip


\medskip

\end{abstract}


\vspace{3cm}
\begin{center}
    \textit{The opinions expressed in the paper are only those of the authors and do not necessarily reflect those of the Autorit\'e de Contr\^ole Prudentiel et de R\'esolution (ACPR).}
\end{center}



\newpage
\section{Introduction}

The behavior of financial institutions, namely banks and insurance companies, constitutes a paradox. On the one hand, they oppose one another in a competition to collect deposits as one may expect for firms in a common sector. In this perspective, the distress of one institution seems good news for the others since there is room for increasing market shares. However, on the other hand, financial institutions need to cooperate. For instance, the withdrawal of a bank from the short term interbank market means that a source of liquidity vanishes, triggering setbacks for other banks. In this case, one financial institution's distress is definitely bad news for the other ones. Thus, even if they are in competition, banks cooperate, insurance companies cooperate and last but not least, banks cooperate with insurance companies. The last point has been ever more significant during the recent years. A support of this cooperation is the interconnections they develop between each other.

In a short-term view, interconnections mirror the resolution of the liquidity needs. As any other firms, banks and insurance companies face asynchronous in-flows and out-flows of cash. One solution is that every institution has its own cash buffer. Alternatively, institutions can create a liquidity pool by sharing their cash to mutualize the liquidity risk \citep{holmstrom1996private, tirole2010theory, rochet2004macroeconomic}. \cite{allen2000financial} explicitly link the interconnectedness of banks to liquidity shocks. Besides the asynchronism of in-flows and out-flows, the liquidity risk is particularly salient since banks are exposed to runs \citep{diamond1983bank} and operate large gross transactions in payment systems \citep{rochet1996controlling}. Indeed flows between institutions are not netted.

However, one may argue that this analysis is not specific to banks and insurance companies since every firm actually faces asynchronous in-flows and out-flows.
 Liquidity concerns are not the only cause of interconnections between financial institutions.
 Moreover, there is evidence in the literature that banks are interconnected not only in the short term but also in the long run. For instance, according to \cite{upper2004estimating}, half of German interbank lending is composed of loans whose maturity is over 4 years (see Figure \ref{FigHauteMaturite}). According to Table 1 in
\cite{Alves2013}, interbank assets with residual maturity larger than one year account for about 50\% of total interbank assets at the European level.\footnote{The existence of long-term interconnections, through loans or shares, is also reported for other countries such as Canada (see Table 3 in \cite{gauthier2012macroprudential}) or France \citep{Fourel2013}.} These long term exposures cannot be explained by a liquidity motive since liquidity is a short term phenomenon. Other possible reasons are horizontal integration (share of a pool of customers via joint products), vertical integration (e.g. risk transfer between insurance and reinsurance companies), ego of top managers aiming at increasing their control of the market and last but not least diversification. Of course, in practice, the network formation stems from a combination of all these motives. However, for reasons explained further, diversification appears as a very important motive. Therefore, in this paper, we consider that these long-term exposures are accounted for by a diversification principle, in a sense that will be defined in the following. \\
\begin{figure}
    \centering
    \includegraphics[scale=0.4]{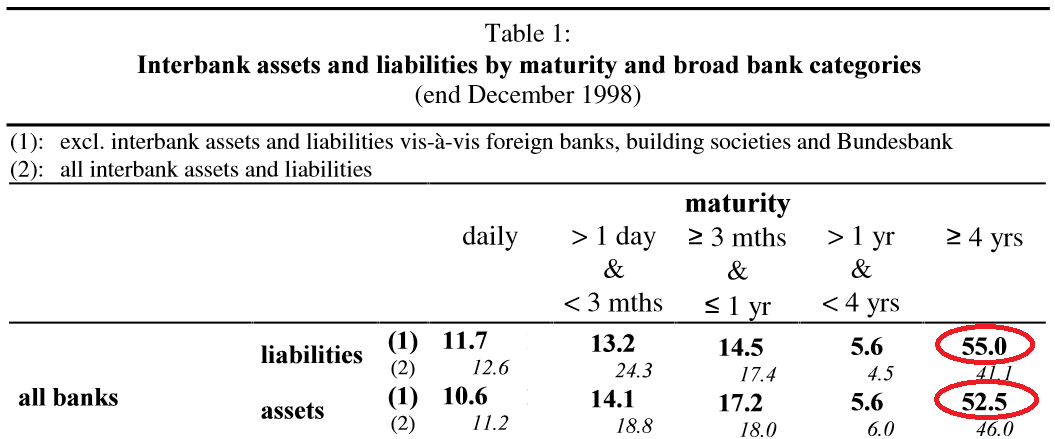}
    \caption{Extract of Table 1 in \cite{upper2004estimating}.}
    \label{FigHauteMaturite}
\end{figure}

The diversification principle is supported by the existence of various business models for banks and insurance companies. The diversity of institutions leads to a diversity of debts and shares available for the other institutions as assets. In the case of insurance companies, there is a clear-cut distinction between mutual funds and profit-oriented insurance companies. The banking sector regroups heterogenous institutions from mutual saving banks to commercial banks. Moreover, the bankassurance business model blurs the separation between banking and insurance activities. This variety can be explained by the different preferences of stakeholders or by historical patterns. Investors who have the same risk aversion gather and form an institution. This heterogeneity can also be linked to a specialization process. For instance, a mutual savings bank funded by farmers is very efficient in granting loans to farmers who in turn favor this bank since it provides the fairest interest rate. This auto-selection mechanism leads to a situation close to a local monopoly. We then understand that for a specific bank, getting interconnected to other institutions is a way to get access to their specific markets. Considering specific markets implicitly assumes that these are not perfectly correlated: for example retail differs from trading. Similarly, insurance companies also specialize in specific risk classes.\footnote{For instance, in the US health insurance sector, specialized institutions exist. The Federal Employees Health Benefits (FEHB) Program is dedicated to federal employees.} Thus being interconnected to other insurance companies allows diversifying one's risk portfolio. All this supports the fact that the diversification principle may explain long-term interconnections among banks and insurance companies. \\

In order to properly model banking and insurance activities, one has to keep in mind that
the banking and insurance sectors are characterized by a very specific production process as well as a heavy regulation. The core activity of a bank consists of the selection of profitable loans and in the management of the resulting maturity transformation. Banks screen potential entrepreneurs for reliable projects and fairly price the resulting interest rate charged. At the same time, they manage the maturity gap: loans to entrepreneurs are long-term assets whereas deposits and issued bonds constitute short-term debt on the liability side.  Information is also key to the core activity of insurance (e.g. damage insurance): the insurer has to efficiently assess the riskiness of the potential policyholder and to deduce the corresponding premium. Strictly speaking the insurance company does not provide maturity transformation. However, its production cycle is reversed: it first collects premia and cushions losses when claims occur. The regulation of the banking and insurance sectors is crucial to maintain people's confidence in the system. In order to avoid bank runs, it is necessary that depositors consider their deposit as safe. Likewise, if policy-holders are not confident in the capacity of the insurer to honor its commitments, they will make other insurance choices. A solvency ratio is imposed to banks and insurance companies: in the case of banks, the ratio compares the riskiness of granted loans with own funds, while for insurance companies, the ratio balances the riskiness of insured risks and the collected premia. \\
Our paper has two main objectives. The first objective is to test whether a diversification motive is a plausible cause for interconnectedness across financial institutions. To do so, we build a model where interconnections are endogenous choices of financial institutions resulting from the maximization of their expected utility. After deriving some theoretical and simulation-based features of the resulting network, we compare these features to those of a stylized financial network (benchmark) based on empirical evidence. The second objective is to fairly assess the impact of regulation on interconnections based on our model. \\

The cornerstone of this paper is the modeling of the endogenous balance sheet of a financial institution, especially interconnections. Endogenous networks have been intensively analyzed in sociology or socio-economics (for a survey, see \cite{goyal2012connections} or \cite{jackson2013economic}). However, finance yields a new field of application. Usually, interconnections among financial institutions are considered as given, especially in applied papers (see among others \cite{cifuentes2005liquidity}, \cite{arinaminpathy2012size}, or \cite{anand2013network}). Endogenous financial networks stem from the seminal paper by \cite{allen2000financial}. For instance, \cite{babus2013formation} models interconnections across banks as the result of an insurance motive: interconnections represent a means of protection against contagion. More recently, \cite{acemoglu2013systemic} focus on the short-term interbank market and model the network formation in the case of risk-neutral banks being able to renegotiate their claims in a case of distress. \cite{elliott2012financial} make a case of showing the incentives that may drive financial network formations. Important insights are brought by this strand of literature inspired by microeconomics and game theory analysis.\footnote{See among others \cite{cohen2011}, \cite{gofman2012}, \cite{Farboodi2014Intermediation} or \cite{Georg2014}.} Nevertheless, in this field, financial institutions only compute their interconnections: the remaining elements of their balance sheet are completely exogenous. This assumption seems suitable in a short-term perspective but not anymore when considering long-term interconnections. Therefore, by including more endogenous balance sheet items than the sole interconnections, we distance ourselves from this strand of research. To the best of our knowledge, the unique paper that considers a complete balance sheet optimization (apart from the debt) is \cite{bluhm2013endogenous}. They propose a dynamic network formation with risk-neutral banks. Using a specific "trial and error" process, the authors first compute the volume of interbank assets (that corresponds to the network's importance) and second its allocation (that corresponds to the network's shape). The allocation is carried out using a matching algorithm based on the strict indifference of banks. In contrast, our paper considers heterogeneously risk-averse banks which explicitly get interconnected to specific counterparts. Last but not least, almost all papers only consider lending (or debt securities) whereas, based on \cite{gourieroux2012bilateral}, our paper also considers shares. This feature cannot be neglected in a long run perspective since financial institutions can take cross-shareholdings. \\

The paper is organized as follows. Section 2 falls into two parts. First, the production process of banks and insurance companies and the regulatory constraints are described. Secondly, we introduce the financial network benchmark. Section 3 presents the theoretical results. After describing the optimization program of financial institutions, we show the existence of an equilibrium and discuss the conditions for its uniqueness. We show that interconnections are usually optimal for financial institutions. These theoretical properties allow to characterize the shape of the network stemming from a diversification motive. Therefore, we compare the shape of a genuine interbank network to a diversification-based one. In Section 4, we first present the computational methodology and the calibration choices. Then we show some simulation results which lead us to assess the proximity of the obtained network to the benchmark network both in terms of balance sheet volume and support of interconnections (debt securities or cross-shareholdings). Section 5 provides an analysis of financial interconnections with respect to financial regulation. Elaborating on  \cite{repullo2013procyclical}, we first show how to fairly analyze interconnectedness and then compare different regulatory frameworks. 
Section 6 concludes. All proofs are gathered in the Appendix.


\section{Balance sheet structure and network benchmark}

In this section, we first describe the economic setup which corresponds to the technology of financial institutions. We introduce the different elements of their balance sheet as well as the regulatory constraints. We then present the stylized network to be later used as a benchmark.

\subsection{Bank and insurance business}

Each bank has access to a specific class of external illiquid assets and each insurance company specializes in one specific class of risk. These classes can be interpreted as main banking (respectively insurance) activities such as, for instance, trading, commercial loans, mortgage loans, sovereign loans (respectively e.g. property insurance, liability insurance, life insurance). \\
The tight relationship between a specific class of assets (respectively risks) and a specific institution has to be interpreted as a consequence of costly portfolio management by investors followed by a specialization process. By portfolio management, we mean the screening process. For banks, that means selecting promising entrepreneurs to finance and offering a fair interest rate. In the case of insurance companies, it means organizing the mutualization of risks, i.e. finding the adequate premium with respect to the policyholder's risk profile. The specialization process strengthens the efficiency of managing a specific portfolio. Due to auto-selection of customers, specialization triggers further specialization.

\subsubsection{Asset side}
Bank $i$'s specific asset book value is labeled $Ax_i$, for $i=1, \dots, n$ (we consider $n$ financial institutions). This asset is some illiquid loan and therefore cannot be exchanged on a market. Thus, no market value can be defined and only its book value is considered in the following.
We denote $R_i$ and $r_i$ the net return of $Ax_i$ and its realization, respectively. The distribution function of the returns $R_1, \dots, R_n$  is denoted $F_R$: $F_R(r_1, \dots, r_n)= \mathds{P} (R_1 \leq r_1, \dots, R_n \leq r_n)$. The corresponding density  is denoted by $f_R$. Banks have access to another external asset, denoted by $A\ell_i$. Its return, deterministic and assumed to be common to all institutions, is denoted $r_{rf}$. Here, $A\ell_i$ is a very liquid and low-risk asset (for instance AAA bonds or S\&P 500 shares), the management of which does not require high technical skills. In the following, $A\ell_i$ will be assimilated to cash, which does not require any screening. We assume that insurance companies' external assets are only composed of $A\ell_i$. Insurance companies are indeed assumed not to have the same capacity of selecting promising innovators as banks, and therefore do not own any specific asset. \\
Besides, Institution $i$ can buy shares or debt  securities issued by Institution $j$ in proportions $\pi_{i,j}$ and $\gamma_{i,j}$, respectively.

\subsubsection{Liability side}
The liability side is composed of equity (that is brought by investors) and nominal debt, whose book values are respectively denoted by $K_i$ and $L_i^{*} $ for Institution $i$. Since equity and debt securities will be traded on the the secondary market, it is necessary to introduce their market values, respectively denoted by $\mathcal{K}_i$ and $\mathcal{L}_i$.

In the case of banks, $L_i^{*} $ includes different types of debts (deposits and bonds of various maturities) considered as homogeneous in terms of seniority.\footnote{For various seniority levels, see \cite{gourieroux2013liquidation}.} Banks issue debt along a common yield curve. In other words, bank debt securities are considered risky (the interest rate curve is above the risk free yield curve) but have a common degree of risk (the same rating, say). Despite this common feature, Bank $i$ chooses its own degree of maturity transformation $\omega_i\in[0,1]$. Let us denote by $T_{Li}$ the average of maturities of all types of debts and by $T_{Ai}$ the maturity of the assets. Then, $\omega_i$ is defined as $\omega_i=1-\dfrac{T_{Li}}{T_{Ai}}$. For instance deposits can be seen as every day re-funded overnight loans by households to banks and therefore their maturity is equal to $0$, yielding $\omega_i=1$. On the opposite, a debt whose maturity equals the asset maturity corresponds to $\omega_i=0$. Banks usually assume that their short-term debt will be rolled over. However, it is not always the case, especially during crises. If a bank is only funded by deposits ($\omega_i=1$), it may happen that all depositors suddenly quit, causing a funding liquidity shock. The same can happen in the case of debt issued with bonds if investors decide not to roll over. In the extreme opposite case ($\omega_i=0$), there is no possible liquidity shock (but there is no maturity transformation).  Banking activity is precisely profitable due to maturity transformation since the interest rate corresponding to long term lending (asset side) is larger than the one corresponding to short term borrowing. In our model, the interest rate charged on the debt of Bank $i$ is deterministic, depends on $\omega_i$ and is denoted by $r_D(\omega_i)$.

In the case of insurance companies, the nominal debt $L_i^{*}$ mostly corresponds to technical provisions relative to the underwritten risks. Therefore, $\omega_i$ can no longer be interpreted as a degree of maturity transformation but as the mean severity of claims. Thus, we do not have necessarily $\omega_i \in [0,1]$ anymore. Contrary to banks, the liability side of an insurer is stochastic. For instance, in line with standard ruin models \citep[see e.g.][]{asmussen2010ruin}, $\omega_i$ could be the parameter of the Pareto distribution in a claims model. Of course, the collected premia directly reflect the risk profile of the insurance contracts.

The balance sheet of Bank $i$ is represented at the initial date and the end date in Tables \ref{TabBalanceSheet0} and \ref{TabBalanceSheet1}, respectively. The dates are represented by an upper-scripted index in parenthesis.
\begin{table}[H]
    \begin{center}
        \begin{tabular}{lcr|lcl}
                 & & Asset                   &  Liability & & \\ \cline{3-4}
                $\left.
                \begin{array}{c}
                    \mbox{interbank} \\
                    \mbox{cross-} \\
                    \mbox{shareholdings}
                \end{array}
                \right.
                $
                &
                $
                \leftrightarrow
                \left\{
                \begin{array}{c}
                    \\
                    \\
                \end{array}
                \right.
                $
                &
                $
                \left.
                \begin{array}{c}
                    \pi_{i,1} \mathcal{K}_1^{(0)}  \\
                    \vdots \\
                    \pi_{i,n} \mathcal{K}_n^{(0)}
                \end{array}
                \right.
                $
                &
                $L_i^*$ & $\leftrightarrow$ & $\mbox{debt}$ \\

                $\left.
                \begin{array}{c}
                    \mbox{interbank} \\
                    \mbox{lending} \\

                \end{array}
                \right.
                $
                &
                $
                \leftrightarrow
                \left\{
                \begin{array}{c}
                      \\
                     \\

                \end{array}
                \right.
                $
                &
                $
                \left.
                \begin{array}{c}
                    \gamma_{i,1} \mathcal{L}_1^{(0)}  \\
                    \vdots \\
                    \gamma_{i,n} \mathcal{L}_n^{(0)}
                \end{array}
                \right.
                $

                &
                $K_i^{(0)}$ & $\leftrightarrow$ & $\mbox{value of the firm}$ \\

                $\mbox{external assets}$ & $\leftrightarrow$ & $Ax_i^{(0)}$ \\

				$\mbox{cash}$ & $\leftrightarrow$ & $A\ell_i^{(0)}$ \\
        \end{tabular}
    \end{center}
	\caption{Balance sheet of Bank $i$ at the initial date $t=0$.}
	\label{TabBalanceSheet0}
\end{table}

\begin{table}[H]
    \begin{center}
        \begin{tabular}{lcr|lcl}
                 & & Asset                   &  Liability & & \\ \cline{3-4}
                $\left.
                \begin{array}{c}
                    \mbox{interbank} \\
                    \mbox{cross-} \\
                    \mbox{shareholdings}
                \end{array}
                \right.
                $
                &
                $
                \leftrightarrow
                \left\{
                \begin{array}{c}
                    \\
                    \\
                \end{array}
                \right.
                $
                &
                $
                \left.
                \begin{array}{c}
                    \pi_{i,1} K_1^{(1)}  \\
                    \vdots \\
                    \pi_{i,n} K_n^{(1)}
                \end{array}
                \right.
                $
                &
                $L_i^{(1)}$ & $\leftrightarrow$ & $\mbox{debt}$ \\

                $\left.
                \begin{array}{c}
                    \mbox{interbank} \\
                    \mbox{lending} \\

                \end{array}
                \right.
                $
                &
                $
                \leftrightarrow
                \left\{
                \begin{array}{c}
                      \\
                     \\

                \end{array}
                \right.
                $
                &
                $
                \left.
                \begin{array}{c}
                    \gamma_{i,1} L_1^{(1)}  \\
                    \vdots \\
                    \gamma_{i,n} L_n^{(1)}
                \end{array}
                \right.
                $

                &
                $K_i^{(1)}$ & $\leftrightarrow$ & $\mbox{value of the firm}$ \\

                $\mbox{external assets}$ & $\leftrightarrow$ & $Ax_i^{(1)}$ \\

				$\mbox{cash}$ & $\leftrightarrow$ & $A\ell_i^{(1)}$ \\
        \end{tabular}
    \end{center}
	\caption{Balance sheet of Bank $i$ at the end date $t=1$.}
	\label{TabBalanceSheet1}
\end{table}
It is important to note that the equity and the debt of the other institutions (on the asset side) must be priced at the market value at $t=0$. At time $t=1$, the book value can be considered.

In line with the Value-of-the-Firm model \citep{merton1974pricing}, the value of debt $L_i$ and equity $K_i$ at any date are linked through the following equilibrium equations
\begin{equation}
\label{K}
    K_i = \max \Big( \sum_{j=1}^n \pi_{i,j}K_j + \sum_{j=1}^n \gamma_{i,j}L_j  + A\ell_i + Ax_i - L^*_i,0   \Big), \mbox{ for } i=1,\dots,n, \mbox{ and }
\end{equation}
\begin{equation}
\label{Lx}
    L_i = \min \Big(  \sum_{j=1}^n \pi_{i,j}K_j + \sum_{j=1}^n \gamma_{i,j}L_j  + A\ell_i + Ax_i, L^*_i  \Big) \mbox{ for } i=1,\dots,n,
\end{equation}
These 2$n$ equations define a liquidation equilibrium. Equation (\ref{K}) corresponds to the simple accounting definition of equity as the net value of assets over debts. Equation (\ref{Lx}) is very similar to \eqref{K} and directly follows from Merton's model: the debt value is the minimum between the asset value and the nominal debt.

Proposition 2 in \cite{gourieroux2012bilateral} states that these equations define a suitable liquidation equilibrium (see Proposition \ref{Prop_Prop2_Gouretal} in Appendix \ref{Appendix_Prop2_Gouretal}).
The cornerstone of our approach will consist in optimizing the balance sheet items of the financial institutions (apart from the equity which is exogenous). Proposition \ref{Prop_Prop2_Gouretal} states that whatever the balance sheet composition of each institution (whatever the  values of $Ax_i$, $A\ell_i$, $\pi_{ij}$, $\gamma_{ij}$ and $L_i^*$ satisfying Assumptions $(A1')$, $(A2')$ and $(A3')$ in Proposition \ref{Prop_Prop2_Gouretal}), the network obtained can theoretically exist (under suitable unique values for $K_i$ and $L_i$, $i=1, \dots, n$). In particular, our optimized network exists and thus the approach we develop in this paper can be carried out.

Note that although \cite{gourieroux2012bilateral} do not consider any maturity, Proposition \ref{Prop_Prop2_Gouretal} still holds true in the presence of $\omega_i$. It is sufficient to replace $L_i^*$ by $L_i^* (1+r_D(\omega_i))$ in the proof.

\subsection{Regulatory constraints}
\label{Constraints}

In line with the usual Basel regulation \citep[see e.g.][Section I]{BCBS2011}\footnote{BCBS means Basel Committee on Banking Supervision.}, the solvency constraint for Institution $i$ is written
\begin{equation}
	K_i^{(0)} \geq k_i^A  Ax_i^{(0)} +  k^{\pi} \Sum_{j=1}^n \pi_{i,j} \mathcal{K}_j^{(0)} + k^{\gamma} \Sum_{j=1}^n \gamma_{i,j} \mathcal{L}_j^{(0)},
\label{Solva_Constraint}
\end{equation}
where $k_i^A$ and $k^{\pi}$ are regulatory parameters (risk weights) for external assets and inter-financial shareholdings and debtholdings, satisfying $0 <k_i^A, k^{\pi}, k^{\gamma} <1$. The parameter relative to the external assets is specific to each institution whereas those relative to interfinancial assets are common within a specific sector (banking or insurance business). This constraint means that the equity must be higher than the risk-weighted assets and aims at ensuring the existence of a sufficient capital buffer to avoid losses for creditors in most cases.

Note that in the case of insurance companies, \eqref{Solva_Constraint} corresponds to the Solvency I regulatory framework \citep[see][]{CEC1979}\footnote{CEC means the Council of the European Communities.}, apart from the term corresponding to interconnections.
Since Solvency II is not implemented so far, we choose not to consider it in our modeling. Moreover, let us emphasize that the weights of banks differ from those of insurance companies. In the case of an insurer, the constraints on $k^{\pi}$ and $k^{\gamma}$ can be relaxed to $0 \leq k_i^{\pi}, k_i^{\gamma} <1$.
\\

Even if we do not focus on liquidity shocks, we introduce a liquidity constraint:
\begin{equation}
	A\ell_i^{(0)} \geq k^L \  l(\omega_i, \    L_i^{*}),
\end{equation}
with $l$ being some increasing function with respect to both variables which will be characterized further and $k_L$ satisfying $0 < k^L < 1$. This constraint aims at ensuring a sufficient liquid assets buffer to face exposure to liquidity risk (maturity transformation in the case of banks and claims in the case of insurance companies) stylized by $\omega_i$ and $L_i$. Note that this constraint is similar to the Basel III Liquidity Coverage Ratio \citep[see][]{BCBS2013}.

\subsection{Summary of the optimization framework}
In short, both banks and insurance companies select their balance sheet items under restrictions (class of assets for banks and class of risks for insurance companies) and regulatory constraints. Their business model is reflected through a size variable and an intensity variable: the size is the total credit granted for a bank and the total of individual risks covered for an insurance company, while the intensity is the degree of maturity transformation for a bank and the claims' severity for an insurance company.

We emphasize that our modeling allows to take the specificities of banks and insurance companies into account in a unified way. The same parameters allow interpretation in terms of banks as well as of insurance companies. However, as we mentioned, the nature of the debt $L_i^*$ and that of the maturity $\omega_i$ are different when considering a bank or an insurance company. In the following, we will mainly focus on banks. 

\subsection{Network Benchmark}
Our testing principle is to compare the network obtained through our modeling and a stylized network, so-called benchmark. In this part, we describe this stylized network along three dimensions. First, we provide the main aggregate items of a bank's balance sheet. Thus, we will be able to check if, apart from interbank assets, the obtained balance sheet composition is close to a real one. Second, we focus on the network shape. This level provides a qualitative assessment of interconnections. Last, the size of interconnections along instruments in a typical banking network is described. This last level provides a quantitative assessment of interconnections. We restrict the analysis to interbank networks in industrial countries, typically the United-States, Canada or Europe. We identify four stylized facts that characterize an interbank network.

\subsubsection{Main aggregate items of a bank's balance sheet}
We consider the Bank Holding Company Performance Report Peer Group Data, published by the Federal Financial Institutions Examination Council, that provides the structure of asset and liability sides for banks above \$10 billion (from 69 banks in 12/2008 to 90 in 12/2012). Figure \ref{FigBHCPRPGData} provides the composition of the asset side and the leverage for these banks. Corresponding informations are summarized in the following stylized fact:
\\

\textbf{Stylized fact 1:} For a typical bank, the external assets ($Ax_i$) represent about 95\% of its total assets while its equity ($K_i$) represents about 5\% of its total assets.\\

\begin{figure}[!h]
    \centering
    \includegraphics[scale=0.45]{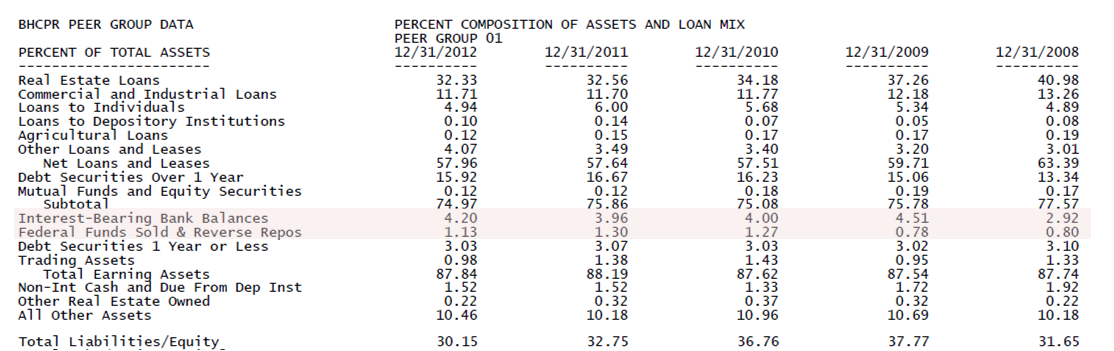} \\
    \small{Comment: interbank assets are mostly concentrated in highlighted lines.}
    \caption{Excerpt of the Bank Holding Company Performance Report Peer Group Data between 12/2008 and 12/2012. Source: www.ffiec.gov.}
    \label{FigBHCPRPGData}
\end{figure}

\subsubsection{Network shape}
National interbank networks\footnote{See \cite{furfine2003interbank} for USA, \cite{wells2002uk} for UK, \cite{upper2004estimating} for Germany, \cite{lubloy2005domino} for Hungary, \cite{van2006interbank} for the Netherlands, \cite{degryse2007interbank} for Belgium, \cite{toivanen2009financial} for Finland, \cite{gauthier2012macroprudential} for Canada, \cite{mistrulli2011assessing} for Italy and \cite{Fourel2013} for France.} are usually characterized by a core-periphery structure \citep{craig2010interbank}. The core is composed of large banks highly interconnected. The periphery is composed of smaller banks which are connected to core banks only. Figure \ref{FigCvPFig1} represents a typical national interbank network. Note that at the international level, the core-periphery structure is much less clear among major banks \citep{Alves2013}. A complete structure seems more representative of the reality. These observations are summarized in the following two stylized facts: \\

\textbf{Stylized fact 2:} For a network composed of banks heterogeneous in size, a core-periphery structure is ideally expected. In other words, matrices $\Mb{\Pi}=(\pi_{ij})_{i,j=1, \dots, n}$ and 
$\Mb{\Gamma}=(\gamma_{ij})_{i,j=1, \dots, n}$ present a block structure with a majority of zeros.\\

\textbf{Stylized fact 3:} For a network composed of large banks homogeneous in size, a complete structure is ideally expected. In other words, $\Mb{\Pi}$ and $\Mb{\Gamma}$ have few zero coefficients.\\

\begin{figure}[!h]
    \centering
    \includegraphics[scale=0.3]{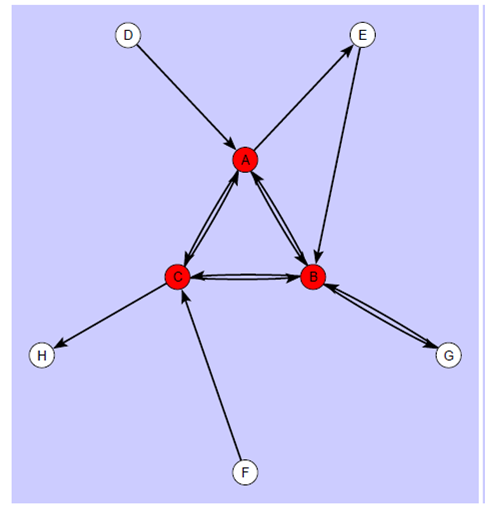} \\
    \caption{Core-Periphery structure. The core is composed of banks A to C while the periphery is composed of banks D to H. Source: Figure 1 in \cite{craig2010interbank}.}
    \label{FigCvPFig1}
\end{figure}

\subsubsection{Interconnections size and support}

As mentioned above, total interbank assets account for about 5\% of total assets. However, data concerning the relative importance of the different instruments are scarce. At the European level (at the end of 2011), according to Table 1 in \cite{Alves2013}, credit claims (direct credit from one bank to another) and debt securities represent 90\% of exposures. The remainder is composed of "other assets". For the 6 largest Canadian banks (as at May 2008), there is a factor 4 between exposure through traditional lending and exposure through cross-shareholdings, as reported in Table 3 in \cite{gauthier2012macroprudential}. \\

\textbf{Stylized fact 4:} In the case of large banks, lending exposures represent a major part of exposures (between 80\% and 90\%).  In other words, $\Mb{\Gamma L^*} \approx \alpha \times (\Mb{\Pi K} + \Mb{\Gamma L^*})$, where $\Mb{K}=(K_i)_{i=1, \dots, n}$,  $\Mb{L^*}=(L_i^*)_{i=1, \dots, n}$ and $\alpha \in [80\%, 90\%]$. However, cross-shareholdings can not be neglected.\footnote{It is paramount to take the relative weight of share securities into account since they are more risky than debt/lending: shareholders lose as soon as the financial institution has losses while a debt holder is only affected if the losses of the financial institution are above its equity. For contagion analysis, cross-shareholdings cannot be neglected.}

\section{Model, theoretical properties and network shape}
\label{TheoreticalProp}

We model the network formation in two steps. The first one -dealt with in this section- concerns the modeling of the behavior of one institution, the state of the others being given. The aim is to determine how a financial institution defines its balance sheet and especially the interconnections knowing the main balance sheet elements of the other ones. For instance, how does a new bank get interconnected to previously existing ones? Or how does a bank adapt its balance sheet to modifications of the structure of others? The second step concerns the whole network formation using the modeling of individual behaviors and will be considered in Section \ref{Networkform}.

Based on the framework introduced in the previous section, a one-period model is built. Banks are risk-averse agents optimizing their balance sheet structure for the shareholder's interest at the initial date $t=0$. The horizon is the final date $t=1$.\\

The assumption that interconnections represent a long-term choice is a cornerstone of our analysis. Interconnections are not motivated by any liquidity features: they correspond to optimal choices in the long-run. Including liquidity-motivated interconnections that stem from daily work of Asset Liability managers, as well as the interactions between short-term and long-term interconnections, constitutes an ongoing work of ours.

A very important concern is the problem of reflexivity: how to technically manage the fact that the choices of financial institutions are interdependent? The main issue is that a complete Nash equilibrium modeling of the whole balance sheet structure -interconnections, external assets and debt- is clearly wishful thinking. It triggers difficulties, especially with respect to privately available information and anticipation formation. Note that in models with Nash equilibrium such as in \cite{babus2013formation} or \cite{acemoglu2013systemic}, choices are only taken at the level of interconnections: all the other components of the balance sheet are exogenous. This scope is arguably adapted in a short-term framework but is clearly unsuitable from a long-term perspective. In order to circumvent a complete game theoretic model, we adopt some simplifying assumptions backed up by practical considerations.

\subsection{Modeling strategy}

We choose an efficient, albeit simple strategy: each financial institution is assuming that the asset side of the other financial institutions is only composed of their external assets. This implies that the institution optimizing its balance sheet is not taking into account the future reactions of the other financial institutions. In this perspective, the optimization program is not strategic. 
Apart from simplifying the resulting optimization program, this strategy corresponds to sound assumptions for each financial institution and this for several reasons.

Firstly, the information set used in the optimization program is very close to the genuinely available one. Actually, bilateral exposures are private information. Publicly available information for any major financial institution are the detailed income statement and the balance sheet. For instance, return-on-asset, return-on-equity, cash, total interbank assets, loans on the asset side, debt and equity on the liability side are easily extracted from the public financial communication of firms or published reports (see Appendix \ref{AppendixExcerptBHC} for an excerpt of the Consolidated Financial Statements for Bank Holding Companies (BHCs) of Bank of America published by the Federal Financial Institutions Examination Council\footnote{http://www.ffiec.gov/nicpubweb/content/help/HelpFinancialReport.htm}). Secondly, note that a large part of debt securities and shares are traded on the secondary market. Therefore, Bank $i$ cannot know exactly who its creditors and shareholders are: Bank $i$ knows its asset side but not the repartition of its liability side. The part of tradable shares is called the floating equity. By analogy, we call the floating debt the part of the debt traded on the secondary market.

Lastly, the absence of anticipation of reaction constitutes an approximation. As previously mentioned, there is no information on bilateral exposures. However, total interbank assets represent about 5\% or 10\% of total assets.\footnote{For instance, on June 30, 2013, the proportion of interbank assets in the total assets is 3.4\% for Bank of America, 13\% for JPM, 8.40\% for Citigroup 8.3\% for Wells Fargo, according to the Consolidated Financial Statements for BHCs.} Each bilateral exposure should be much smaller: 0.5\% of total assets seems a reasonable upper bound. Therefore, when a new bank gets interconnected, the new interconnections do not significantly modify its balance sheet. It may trigger a reaction from its own counterparts but the effects can be neglected by comparison to the risk borne in the external assets for instance. As we will see in the simulation results, the reaction of counterparts only has a light influence on each institution, leading to a rapid stabilization of the network. This provides an indication that this assumption of absence of anticipation can be accepted as a first step towards building more realistic models.

Then this assumption allows us  to derive in the next subsection some strong and tractable theoretical results.

\subsection{Optimization program}
Bank $i$ is managed for the benefits of its investors (i.e. shareholders) who are risk-averse and endowed with an initial capital $K_i^{(0)}$. The risk-aversion of the investors of Bank $i$ is represented by a utility function $u_i$. 
We denote $1-c^{\pi}_j$ (respectively $1-c^{\gamma}_j$) the floating equity (respectively debt) of Bank $j$, for $j=1, \dots, n$. \\
In line with our modeling strategy, we scale the total assets of Bank $j$ by $\kappa_j = 
\dfrac{L_j^{(0)} + K_j^{(0)} }{ Ax_j^{(0)} + A\ell_j^{(0)}}$. 
These scaling factors compensate for the fact that we consider that the counterparts are not interconnected. Thus, we get the following approximation for the equity of Bank i at time $t=1$:
\begin{align}
K_i^{(1)}
& = \max \Big[ Ax_i^{(1)} + A\ell_i^{(1)}+ \Sum_{j=1}^n \pi_{i,j} \max \left( \kappa_j(Ax_j^{(1)} + A\ell_j^{(1)}) - L_j^{*(1)}, 0 \right)
\nonumber \\
& + \Sum_{j=1}^n \gamma_{i,j}\min \left( \kappa_j(Ax_j^{(1)} + A\ell_j^{(1)}) , L_j^{*(1)} \right) - [1+r_D(\omega_i)] L_i^{(0)},0 \Big].
\end{align}
If we denote by $\mathds{E}_0$ the expectation computed at time $t=0$, the optimization program $\mathcal{P}_i$ of Bank $i$ is
$$
	\mathcal{P}_i := \left\{
	\begin{array}{rlr}
       \max                          & \mathds{E}_0 \left[ u_i \left( K_i^{(1)} \right) \right] 		\\
       Ax_i^{(0)},A\ell_i^{(0)} 											\\
       L_i^{(0)},\omega_i 													\\
       \pi_{i,1},\dots,\pi_{i,n} 											\\
       \gamma_{i,1},\dots,\gamma_{i,n} 										\\
       \mbox{such that (s.t.)}							 &	Ax_i^{(0)} + A\ell_i^{(0)}
											+ \Sum_{j=1}^n \pi_{i,j} \mathcal{K}_j^{(0)} + \Sum_{j=1}^n \gamma_{i,j} \mathcal{L}_j^{(0)}  = K_i^{(0)} + L_i^{(0)} & (BC) \\
       								 &  K_i^{(0)} \geq k_i^A \ Ax_i^{(0)} + k^{\pi} \Sum_{j=1}^n \pi_{i,j} \mathcal{K}_j^{(0)} + k^{\gamma} \Sum_{j=1}^n \gamma_{i,j} \mathcal{L}_j^{(0)}	& (SC) \\
       								        								 	
       								 & A\ell_i^{(0)} \geq k^L \ l(\omega_i, \  L_i^{(0)}) & (LC) \\
       								 & Ax_i^{(0)} \geq 0, A\ell_i^{(0)} \geq 0, L_i^{(0)} \geq 0 \\
       								 & \omega_i \in [0,1] \\
       								 & \forall j\in \{1, \dots, n\}, 0 \leq \pi_{i,j} \leq 1 - c^{\pi}_j , 0 \leq\gamma_{i,j} \leq 1 - c^{\gamma}_j.
    \end{array}
    \right.
$$
The constraint $(BC)$ ensures the balance sheet equilibrium at the initial date. Note that this constraint allows the network resulting from our formation process (see Section \ref{Networkform}) to satisfy \eqref{K} for each institution. The inequalities $(SC)$ and $(LC)$ are respectively the regulatory solvency and liquidity constraints presented in Section \ref{Constraints}. $BC$, $SC$ and $LC$ stand for Balance sheet Constraint, Solvency Constraint and Liquidity Constraint, respectively.

\subsection{Solution analysis}

We define the position $P_i$ of Bank $i$ as the difference between its total assets (denoted by $A_i$) and its nominal debt. Therefore, at time $t=1$, $P_i^{(1)} = A_i^{(1)} - L_i^{*(1)}$. If this difference is positive, the position is simply the equity; if the difference is negative, the position is the loss for creditors (while the equity is equal to zero in this situation). $P$ can be interpreted as the profit-and-loss. 

The uniqueness of the solution usually requires the strict concavity of the objective function. The concavity of $u_i \circ K_i^{(1)}$ (where $\circ$ denotes the composition operator) is not a necessary condition since we could expect that the integration operation makes the expectation strictly concave even if $u_i \circ K_i^{(1)}$ is not strictly concave everywhere (see the Appendix for more details). Moreover, it would impose conditions on $F_R$. Thus, we look for conditions on $u_i \circ K_i^{(1)}$. Due to their limited liability, shareholders aim at maximizing the expected utility of the equity. The latter is defined as $K_i^{(1)}=\max(P_i^{(1)}, 0)$, making $u_i \circ K_i^{(1)}$ non-differentiable and introducing a level shape. An unfortunate consequence is that for standard utility functions $u_i$, $u_i \circ K_i^{(1)}$ is not strictly concave and not even concave. Then our strategy is to approximate the real equity by a function $v(P_i^{(1)})$ to obtain the concavity. From an economic perspective, it is satisfactory  to consider a transformation of the equity, as we will see in the following. Therefore, we decompose the analysis of $\mathcal{P}_i$ into two steps. Firstly, we show that under mild assumptions there exists a solution (Theorem \ref{PropExistenceSolOptiProSol}). Secondly, we transform the optimization program $\mathcal{P}_i$ into a close one ($\mathcal{P}_i'$) for which existence and uniqueness are ensured (Theorem \ref{PropExistenceUniciteSolOptiProSol}).

\subsubsection{Analysis of the exact optimization program}
Contrary to usual optimization programs where the total wealth is exogenous, increasing wealth by issuing debt is allowed in $\mathcal{P}_i$. Therefore, intuitively, the main difficulty in showing the existence of a solution is to show that Bank $i$ has no gain in issuing an infinite amount of debt. The argument is as follows. The equity is exogenously fixed. Therefore, $(SC)$ implies that the total value of risky assets is bounded. Thus, starting from a specific amount of debt, the funding obtained by issuing more debt is necessarily invested in the risk free liquid asset. But since the interest rate charged on the debt is higher than the risk free rate, it is not profitable to issue debt to invest in liquid assets. In other words, banks are expected to invest in risky assets: granting credit is the core activity of banks. \\
All this goes to state the following proposition:
\begin{Th}[Existence of a solution to $\mathcal{P}_i$]
If
\begin{itemize}
	\item[$\bullet$] $(A1)$ the investors neglect interconnections among their counterparts;
	\item[$\bullet$] $(A2)$ the utility function $u_i$ is continuous and strictly increasing;
	\item[$\bullet$] $(A3)$ the distribution function $F_R$ is continuous. Moreover, the density $f_R$ is strictly positive on $[a, + \infty)^n$, for some $a \in \mathds{R}$;
	\item[$\bullet$] $(A4)$ the yield curve, $r_D$, is continuous and strictly higher than the risk free rate;
\end{itemize}
then there exists a solution to $\mathcal{P}_i$.
\label{PropExistenceSolOptiProSol}
\end{Th}


Assumption $(A1)$ is both a technical assumption and a way to reflect the restricted information available for each agent. Assumptions $(A2)$, $(A3)$ and $(A4)$ are very common in the literature and not restrictive. 

\subsubsection{Analysis of the approximated optimization program}
As stressed before, it appears impossible to establish the uniqueness for $\mathcal{P}_i$ except in particular cases of simple models for $F_R$. We therefore consider an optimization problem $\mathcal{P}_i'$ where the sole difference with $\mathcal{P}_i$ is that the objective function is the expected utility of a strictly increasing transformation (denoted by $v$) of the position of Bank $i$, $P_i^{(1)}$. Considering the position directly makes things easier. However, it means not taking into account the limited liability which has some important implications. Indeed, it plays the role of a protection against extreme events for the managers: they are impacted by regular shocks but not by extreme ones. Some phenomena cannot be explained by macro-economic models ignoring limited liability.
The optimization program $\mathcal{P}_i'$ is
$$
	\mathcal{P}_i' := \left\{
	\begin{array}{rlr}
       \max                          & \mathds{E}_0\left \{ u_i\left[ v( P_i^{(1)}) \right] \right \} 		\\
       Ax_i^{(0)},A\ell_i^{(0)} 											\\
       L_i^{(0)},\omega_i 													\\
       \pi_{i,1},\dots,\pi_{i,n} 											\\
       \gamma_{i,1},\dots,\gamma_{i,n} 										\\
       \mbox{s.t.}							 &	Ax_i^{(0)} + A\ell_i^{(0)}
											+ \Sum_{j=1}^n \pi_{i,j} \mathcal{K}_j^{(0)} + \Sum_{j=1}^n \gamma_{i,j} \mathcal{L}_j^{(0)}  = K_i^{(0)} + L_i^{(0)} & (BC) \\
       								 &  K_i^{(0)} \geq k_i^{A} \ Ax_i^{(0)} + k^{\pi} \Sum_{j=1}^n \pi_{i,j} \mathcal{K}_j^{(0)} + k^{\gamma} \Sum_{j=1}^n \gamma_{i,j} \mathcal{L}_j^{(0)}	 & (SC) \\
       								        								 	
       								 & A\ell_i^{(0)} \geq k^L \ l(\omega_i, \  L_i^{(0)}) & (LC) \\
       								 & Ax_i^{(0)} \geq 0, A\ell_i^{(0)} \geq 0, L_i^{(0)} \geq 0 \\
       								 & \omega_i \in [0, 1] \\
       								 & \forall j\in \{1, \dots, n \}, 0 \leq \pi_{i,j} \leq 1 - c^{\pi}_j , 0 \leq\gamma_{i,j} \leq 1 - c^{\gamma}_j.
    \end{array}
    \right.
$$

With this specification, the level aspect of the limited liability is removed and the transformation $v$ ensures flexibility. For instance, with $v=Id$, one considers the usual maximization of the expected utility of profits. Alternatively, $v$ can be chosen to closely fit the design of the limited liability of shareholders while relaxing their complete indifference for loss magnitude. In the latter case, $\mathcal{P}_i'$ is very close to $\mathcal{P}_i$. 

In short, the argument for the existence of a solution of $\mathcal{P}_i'$ is similar to the argument for the existence of a solution of $\mathcal{P}_i$. The uniqueness mainly stems from the strict concavity of the objective function we obtain by adjusting $v$. However, the strict convexity of the constraints is necessary, imposing restrictions on the functional form of $(LC)$ (see the proof for details). The following theorem provides the result regarding uniqueness:

\begin{Th}[Existence and uniqueness of a solution to $\mathcal{P}_i'$]~\\
    Under $(A1)$, $(A2)$, $(A3)$, $(A4)$ and the extra assumptions:
    \begin{itemize}
    	\item[$\bullet$] $(A5)$ the composition of the transformation function $v$ and the utility function $u_i$ is strictly concave: $\forall P \in \mathds{R}, (u_i \ o \ v)''(P)<0$; 
    	\item[$\bullet$] $(A6)$ the  interest rate on debt is strictly concave: $\forall \omega_i \in [0,1], r_D''(\omega_i) < 0$;
    	\item[$\bullet$] $(A7)$ the interest rate on debt satisfies $\forall \omega_i \in [0,1], r_D'(\omega_i) \ne 0$; 
    	\item[$\bullet$] $(A8)$ the function $l$ in $(LC)$ satisfies
    	$$ \frac{\partial^2 l}{\partial \omega_i^2} \geq 0 \ \mbox{ and }\  \frac{\partial^2 l}{\partial \omega_i^2} \frac{\partial^2 l}{\partial {L_i^{(0)}}^2} \geq \left( \frac{\partial^2 l}{\partial \omega_i \partial L_i^{(0)}} \right)^2;$$
    \end{itemize}
     there exists a unique solution to $\mathcal{P}_i'$
    in the following sense. If all control variables appearing on the asset side of Bank $i$ are fixed apart from one variable, denoted by $Ac_i^{(0)}$, then there is uniqueness of the triplet ($Ac_i^{(0)}$,~ $L_i^{(0)}$, $\omega_i$).
    \label{PropExistenceUniciteSolOptiProSol}
\end{Th}

Note that the result of Theorem \ref{PropExistenceUniciteSolOptiProSol} is equivalent to saying that the main balance sheet items are unique. Indeed, the value of total assets $A_i^{(0)}$, the degree of maturity transformation $\omega_i$ and the debt $L_i^{(0)}$ are unique.
Due to the high number of control variables on the asset side and the complexity of the problem, it seems impossible to prove the uniqueness of all control variables (see the Appendix for more details). 
The uniqueness for all control variables will be verified on simulations.

\subsubsection{Approximation properties}
As mentioned before, the transformation function $v$ gives room for flexibility. Lemma \ref{CorV} provides two specifications satisfying $(A5)$, corresponding respectively to the position and a very good approximation of the equity.

\begin{Lem}[Some specifications of $v$ and $u_i$] ~~
    \begin{itemize}
        \item \textbf{i)} If $\forall P \in \mathds{R}, v(P)=P$, then $(A5)$ reduces to $u_i'' <0$.
        \item \textbf{ii)} If $\forall P \in \mathds{R}, v(P)=\log \left( \exp(P)+1 \right)$, then $(A5)$ is satisfied for the utility function $u_i=\log$.
    \end{itemize}
    \label{CorV}
\end{Lem}

The approximation corresponding to $v(P)=\log \left( \exp(P)+1 \right)$ is shown in Figure \ref{Approximations_Equity ii)}. As we can see, the approximation error is very low. In the perspective of maximizing the utility, this function is probably even more satisfactory than the real equity. Indeed the utility of the equity is equal to zero whatever the position if the position is negative. In reality, one may think that the bank's managers prefer a light insolvency situation to a large one, for example for the sake of reputation. It is be difficult to find funding to build a new project after letting an institution in a state of large insolvency. Our approximation function is strictly increasing and therefore takes this aspect into account. This is especially true for position values not too far away from the insolvency point.
\begin{figure}[H]
    \center
    \includegraphics[scale=0.65]{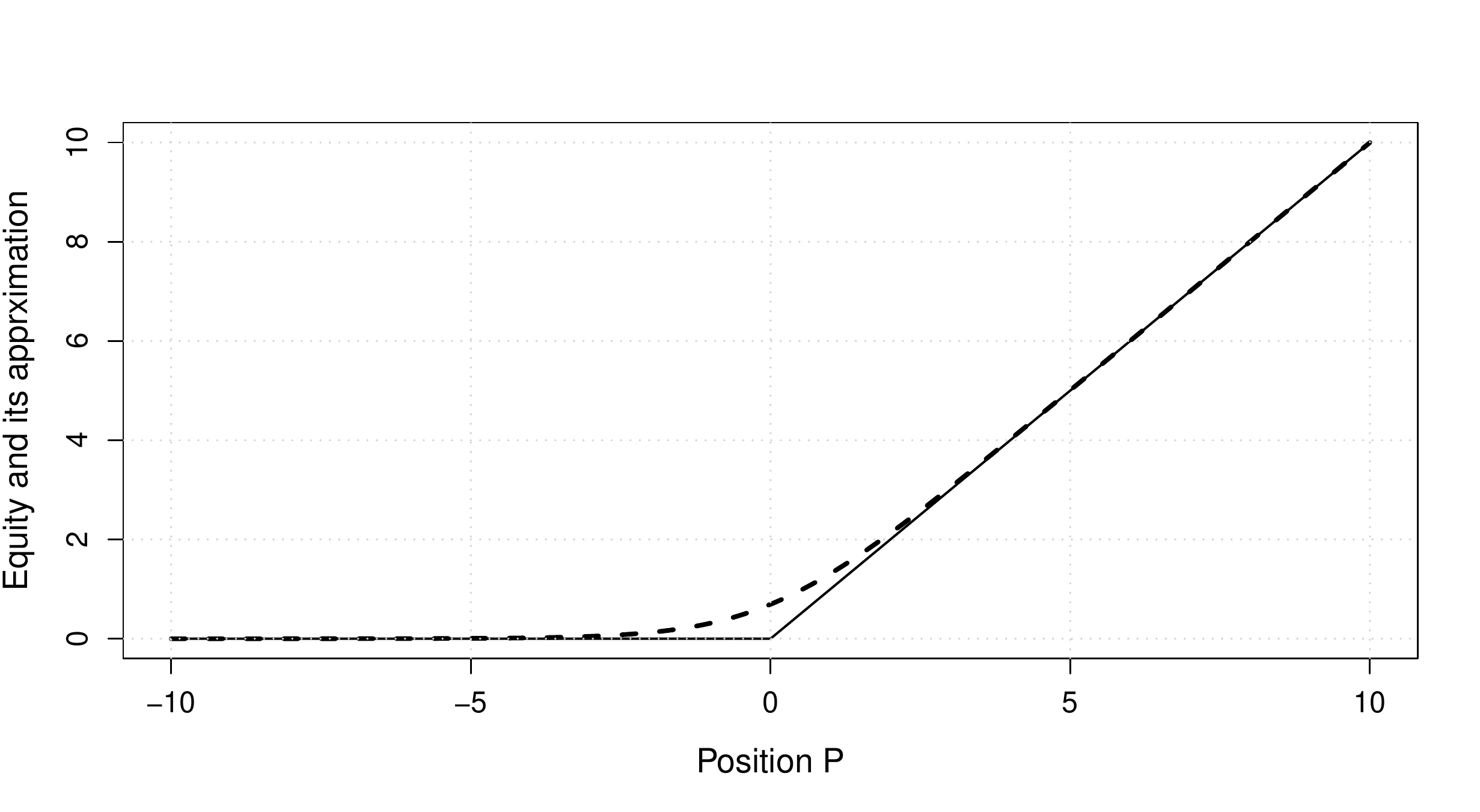}
    \caption{The solid line represents the equity and the dashed line displays the approximated equity using funtion $v(P)=\log (\exp(P)+1)$.}
    \label{Approximations_Equity ii)}
\end{figure}
Lemmas \ref{CorrD} and \ref{CorL} provide a specification for the interest rate curve $r_D$ and the function $l$ appearing in $(RC)$, respectively satisfying $(A6)$ and $(A8)$.

\begin{Lem}[Specification of function $r_D$]~\\
An interest rate curve of the form
\begin{equation}
    r_D(\omega) = \alpha  - \beta \exp (\omega), ~~~\mbox{ for } \omega \in [0,1],
\end{equation}
satisfies $(A6)$.
\label{CorrD}
\end{Lem}

\begin{Lem}[Specification of function $l$]~\\
    The function defined by 
    $$l(\omega,L)=\exp(\omega) \exp(L), \ \ \ \mbox{ for } \omega \in [0,1] \mbox{ and } L \in \mathds{R}_+,$$
     satisfies $(A8)$.
    \label{CorL}
\end{Lem}

\subsubsection{Choice}

Previous theoretical results provide different suitable specifications (especially of the function $v$) leading to a unique solution of the optimization program. In order to clarify the presentation, let us make a clear recommendation of choice. The following result is directly derived from Theorem \ref{PropExistenceUniciteSolOptiProSol} and Lemmas \ref{CorV}, \ref{CorL} and \ref{CorrD}.

\begin{Cor}[Existence and uniqueness to a solution of a specific optimization program] ~\\
    Additionally to $(A1)-(A4)$, let us consider:
    \begin{itemize}
        \item[$\bullet$] a logarithmic utility function
                                $$ u_i(x) = \log(x), \ \ \ \mbox{ for } x \in \mathds{R};$$
        \item[$\bullet$] the following approximation of the limited liability of shareholders:
                                $$ v(P) = \log \left( \exp(P) + 1 \right), \ \ \ \mbox{ for } P \in \mathds{R};$$
        \item[$\bullet$] the following liquidity constraint:
                               $$ l(\omega,L)=\exp(\omega) \exp(L), \ \ \ \mbox{ for } \omega \in [0,1] \mbox{ and } L \in \mathds{R}_+;$$
        \item[$\bullet$] the following interest rate curve:   $$ r_D(\omega) = \alpha  - \beta \exp (\omega), ~~~\mbox{ for } \omega \in [0,1]. $$
    \end{itemize}
    Then, the associated optimization program $\mathcal{P}_i'$ has a unique solution.
    \label{PropChoix}
\end{Cor}

To conclude this section, let us emphasize that all parameters and variables required to perform the optimization can be obtained via publicly available data.

\subsection{Optimal interconnections}
\label{Chapnetworks_Subsec_Optimal}

Previous theoretical results ensure that the bank's maximization program has a (unique) solution. However, we did not characterize this solution, in particular the interconnections. In this part, we show that under some conditions, it is optimal for a bank to get interconnected. In this section, in order to simplify the presentation and to explain the main features, we do not take into account the control variables $A\ell_i$ and $\omega_i$, as well as the liquidity constraint $(LC)$.

In order to start, let us consider a simplified case of a portfolio composed of a quantity $Ax$ and a quantity $\pi$ of assets having respectively random variables $R_g$ and $R_g^{\pi}$ as gross returns, under a solvency constraint.\footnote{For the sake of simplicity, the product $\pi \mathcal{K}$ of the complete program has been simplified into $\pi$.} The penalization weights are respectively $k^A$ and $k^{\pi}$. The corresponding optimization program is
$$
	\mathcal{P_{RA}} := \left\{
	\begin{array}{rlr}
       \max                          & \mathds{E} \left[ u( Ax R_g + \pi R_g^{\pi} ) \right]		\\
       Ax,\pi										\\
       \mbox{s.t.}							 &	k^A Ax + k^{\pi} \pi \leq 1 \\
       								 &  Ax \geq 0 \\
       								 & 0 \leq \pi \leq 1 \\
    \end{array}
    .
    \right.
$$
The Karush, Kuhn and Tucker (KKT) Theorem \citep{karush1939minima, kuhn1951proceedings} allows to derive the following proposition.
\begin{prop}
\label{Prop_PositivenessPi}
For the sake of simplicity, we denote $f=\mathds{E}  \left[ u( Ax R_g + \pi R_g^{\pi} ) \right]$.
Under the condition
$$ \forall Ax \in \mathds{R}_+ \mbox{ and } \forall \pi \in [0,1], \dfrac{\dfrac{\partial f}{\partial Ax}(Ax, \pi)}{k^{A}} < \dfrac{\dfrac{\partial f}{\partial \pi}(Ax,\pi)}{k^{\pi}},$$
the optimal $\pi^*$ is different from 0.
\end{prop}
This shows that under the condition that the derivative of the expected utility with respect to $\pi$ (relative to its corresponding weight) is higher than the one with respect to $Ax$, the optimal $\pi^*$ is strictly positive. Proposition \ref{Prop_PositivenessPi} does not provide the solution but gives an indication that interconnections can be strictly positive under some conditions. This result can be generalized to a higher number of assets. Note that this illustrative program does not contain any equality constraint. However, such a constraint can be trimmed by replacing one control variable in function of the others. That reduces the problem's dimension. This point will be further detailed in the following.

Due to the high complexity of our optimization problem (high dimension and high number of constraints), the KKT conditions are very numerous and therefore it seems impossible to derive the solution in a closed form. We decompose the analysis 
in different steps. We first consider a risk-neutral agent maximizing the value of its portfolio without limited liability. Secondly, we consider the case of a risk-averse agent and finally the limited liability is taken into account.\\

\noindent \textbf{Risk-neutral agent without limited liability:}\\
In the risk-neutral case, the utility function is the identity function. 
Therefore, we can consider the following optimization program:
$$
	\mathcal{P_{RN}} := \left\{
	\begin{array}{rlr}
       \max                          & \left( Ax \mathds{E}(R_g) + \pi \mathds{E}(K^{(1)}) \right)		\\
       Ax,\pi										\\
       \mbox{s.t.}							 &	k^A Ax + k^{\pi} \pi \ \mathcal{K}^{(0)} \leq 1 \\
       								 &  Ax \geq 0 \\
       								 & \pi \geq 0 \\
    \end{array}
    \right.
    ,
$$
where $K^{(1)}$ is the equity value (book value) of another institution at time $t=1$ and $\mathcal{K}^{(0)}$ is the equity value (market value) of this institution at time $t=0$. 

By using the same type of argument as in Proposition \ref{Prop_PositivenessPi}, it is easy to show that if $\mathds{E}(R_g)>0$ or $\dfrac{\mathds{E}(K^{(1)})}{\mathcal{K}^{(0)}}>0$, then
\Bit
\item if $\dfrac{\mathds{E}(R_g)}{k^A} > \dfrac{\mathds{E}(K^{(1)})}{k^{\pi} \mathcal{K}^{(0)}}$, the unique solution is $\left( Ax^*=\dfrac{1}{k^A}, \pi^*=0 \right)$;
\item if $\dfrac{\mathds{E}(R_g)}{k^A} < \dfrac{ \mathds{E}(K^{(1)}) }{k^{\pi} \mathcal{K}^{(0)}}$, the unique solution is $\left( Ax^*=0, \pi^*=\dfrac{1}{k^{\pi} \mathcal{K}^{(0)}} \right)$;
\item if $\dfrac{\mathds{E}(R_g)}{k^A} = \dfrac{ \mathds{E}(K^{(1)}) }{k^{\pi} \mathcal{K}^{(0)}}$, the solution is not unique.
\Eit
Therefore, due to the solvency constraint, a risk-neutral agent only invests in the asset having the highest return with respect to its specific regulatory weight in the solvency constraint.

Let us now consider the case where a limit to the availability is introduced: the constraint $\pi \geq 0$ is replaced by $0 \leq \pi \leq c$.
In this case, if $\dfrac{\mathds{E}(R_g)}{k^A} < \dfrac{\mathds{E}(K^{(1)})}{k^{\pi} \mathcal{K}^{(0)}}$, $\pi^*=\min \left( c , \dfrac{1}{k^{\pi} \mathcal{K}^{(0)}} \right)$. Therefore, if $c < \dfrac{1}{k^{\pi} \mathcal{K}^{(0)}}$, investing all in $\mathcal{K}^{(0)}$ does not bind the solvency constraint. In this case (and if $\mathds{E}(R_g)>0$), an investment in $Ax$ completes the portfolio.
This result can be easily generalized to the case of $n$ institutions and where it is possible to invest in the debt $L_j, j=1, \dots, n,$ of the other institutions. This is done in the following theorem.
\begin{Th}
\label{Prop_ProgramRNG}
Let us consider the following optimization program:
$$
	\mathcal{P_{RNG}} = \left\{
	\begin{array}{rlr}
       \max                          &  \left( Ax_i \mathds{E}(R_{g,i}) +  \sum_{j=1}^n \pi_{ij} \mathds{E}(K_j^{(1)}) + \sum_{j=1}^n \gamma_{ij} \mathds{E}(L_j^{(1)}) \right)	\\
       Ax_i,\pi_{ij}, \gamma_{ij}										\\
       \mbox{s.t.}							 &	k^A Ax_i + k^{\pi} \sum_{j=1}^n \pi_{ij} \mathcal{K}_j^{(0)} + k^{\gamma} \sum_{j=1}^n \gamma_{ij} \mathcal{L}_j^{(0)} \leq 1 \\
       								 &  Ax \geq 0 \\
       								 & 0 \leq \pi_{ij} \leq c^{\pi} \\
       								 & 0 \leq \gamma_{ij} \leq c^{\gamma} \\
    \end{array}
    .
    \right.
$$
To find this problem's solution, let us sort in decreasing order the following returns (relative to their penalty weight):
$\dfrac{ \mathds{E}(R_{g,i}) }{k^A}$, $\dfrac{ \mathds{E}(K_j^{(1)}) }{k^{\pi} \mathcal{K}_j^{(0)}}$ ($j=1, \dots, n$), $\dfrac{ \mathds{E}(L_j^{(1)})}{k^{\gamma} \mathcal{L}_j^{(0)}}$ ($j = 1, \dots, n$).
The optimal solution consists in investing as much as possible in the asset having the highest return with respect to its regulatory weight. When this asset is not available anymore, it is better to invest as much as possible in the second one, and so on. This is repeated until the solvency constrained is binding.
\end{Th}

\noindent \textbf{Risk-averse agent without limited liability} \\
A risk-averse agent aims at decreasing the variance of its portfolio. To this purpose, it is necessary to diversify. Therefore, in this case, we can expect an investment in many assets, contrary to the "binary" investment described previously. This is confirmed by numerical experiments.
\\

\noindent \textbf{Agent with limited liability} \\
In the previous considerations, we did not take into account the limited liability as well as the fact that equity and debt have very different features. Therefore, we could not see the implications of the fact that the $\pi_{ij}$ and the $\gamma_{ij}$ are related to very different instruments. To pinpoint these implications, let us consider a stylized set-up with two banks. One can identify four situations in which Bank 1 (or 2) is either solvent or in default. Table \ref{TableInstitutionsState} reports these 4 states. Let us focus on the impact of limited liability for Bank 1. We assume that Bank 1 builds interconnections with Bank 2 anyway (for example in order to reduce its variance) and we discuss the distribution among shares and debt securities.
\begin{table}[H]
    \centering
    \begin{tabular}{|c|c|c|}
    \hline
           & Bank 2 in default & Bank 2 solvent \\
           \hline
       Bank 1 in default & $e_{11}$ & $e_{12}$ \\
       \hline
       Bank 1 solvent & $e_{21}$ & $e_{22}$ \\
        \hline
    \end{tabular}
    \caption{Banks' states.}
    \label{TableInstitutionsState}
\end{table}
The expected utility of Bank 1 is written as follows
$$ \mathds{E}(U_1) = \mathds{P}(e_{11}) \ PO(e_{11}) + \mathds{P}(e_{12}) \ PO(e_{12}) + \mathds{P}(e_{21}) \ PO(e_{21}) + \mathds{P}(e_{22}) \ PO(e_{22}), $$
where $\mathds{P}(e)$ is the probability of being in state $e$ and $PO(e)$ the associated payoff for Bank 1. Due to limited liability, $PO(e_{11})=PO(e_{12})=0$.
Thus
$$ \mathds{E}(U_1) = \mathds{P}(e_{21}) \ PO(e_{21}) + \mathds{P}(e_{22}) \ PO(e_{22}). $$
In the state $e_{21}$, Bank 2 defaults, meaning that its equity is equal to zero. It is therefore more interesting to invest in its debt. In the state $e_{22}$, Bank 2 is solvent. Thus, if the equity of Bank 2 has a higher return than its debt with respect to their regulatory weights, Bank 1 prefers investing in the share securities of Bank 2, thus increasing the $\pi_{12}$.
If the correlation $\rho$ between the external assets of both banks is highly positive, both banks are likely to be solvent and to default simultaneously. That means that $ P(e_{21})$ is very low, giving:   $\mathds{E}(U_1) \approx \mathds{P}(e_{22}) \ PO(e_{22})$. In this situation, Bank 1 prefers investing in share securities.
On the contrary, if the correlation $\rho$ between the external assets of both Banks is highly negative, Bank 2 is likely to default when Bank 1 is solvent. In this case $\mathds{E}(U_1) \approx \mathds{P}(e_{21}) \ PO(e_{21})$ and Bank 1 prefers investing in debt securities.

It is important to understand that the asymmetry between the cases $\rho>0$ and $\rho<0$ is due to the limited liability feature. Indeed, let us assume that Bank 1 has no limited liability and thus is not indifferent to losses. If $\rho$ is highly positive, $\mathds{E}(U_1) \approx \mathds{P}(e_{11}) \ PO(e_{11}) + \mathds{P}(e_{22}) \ PO(e_{22})$. In state $e_{11}$, Bank 2 defaults and it is better to invest in its debt whereas in state $e_{12}$, it is better to invest in its shares. Therefore, it can be appropriate to invest in both instruments and thus the asymmetry disappears. The same happens for a highly negative $\rho$.

However, keep in mind that this set-up is too minimal to show all the implications of the limited liability.

\subsection{Cost of funding}
In the considerations of Section \ref{Chapnetworks_Subsec_Optimal}, we assumed that the agent owns a sufficient amount of wealth to invest until the solvency constraint is binding. However, the capital $K_i^{(0)}$ is  very low compared to the total assets to invest (due to the regulatory weight values).
Thus, once the total capital has been used, the institution must raise debt in order to continue to invest. Returns of shares and debt securities must be netted by the cost of funding. To make the investment attractive (in terms of net returns), the cost of raising debt should be lower than the returns of shares and debt securities. 

Let us now state some results about the returns of investments in shares and debt securities issued by other institutions, compared to their funding cost.
For the sake of simplicity of the interpretation, before stating the result for general functions $u_i$ and $v$, we propose a result in the case where $u_i$ and $v$ are the identity functions. It corresponds to the case of a risk-neutral institution maximizing its position $P_i^{(1)}$.

\begin{prop}[Returns against opportunity cost, in the case of a risk-neutral institution maximizing the expectation of its position] ~\\
    \label{PropPositiveness_1}

    \begin{itemize}
        \item[$\bullet$] The expected return of a share issued by Bank $j$ is larger than the cost of funding of Bank $i$ if and only if
                            \begin{equation}
                            \label{SRN}
                                 \int_{\frac{-b_j}{a_j}}^{+ \infty} (a_j + b_j r_j) f_{R,j}(r_j) \  dr_j > [1+r_D(\omega_i)]\  \mathcal{K}_j^{(0)},
                            \end{equation}
        where $a_j=\kappa_j Ax_j^{(0)}$, $b_j=\kappa_j \left( Ax_j^{(0)}+ A\ell_j^{(0)} (1+r_f) \right) - L_j^{*} [1+r_D(\omega_j)]  $ and $f_{R,j}$ is the marginal density of the net return of the external asset of Bank $j$.
        \item[$\bullet$] The expected return of the debt issued by Bank $j$ is higher than the cost of funding of Bank $i$ if and only if
$$                            \label{DRN}
                                \int_{- \infty}^{\frac{L_j^*[1+r_D(\omega_j)]-b_j}{a_j}} (a_j r_j + b_j) f_{R,j}(r_j) \  dr_j  + L_j^{*}  c_j [1+r_D(\omega_j)] >  [1+r_D(\omega_i)] \mathcal{L}_j^{(0)},
$$
                            where $c_j=\mathds{P} \left( r_j > \frac{L_j^{*(1)}-b_j}{a_j} \right).$
    \end{itemize}
\end{prop}

\begin{prop}[Returns against opportunity cost, in the general case of an institution maximizing the expectation of the utility of its equity]~\\

\Bit        
\item The expected return of a share issued by Bank $j$ is larger than the cost of funding of Bank $i$ if and only if 
$$
 \int_{-\frac{b_j}{a_j}}^{+ \infty} (a_j r_j + b_j) \ w(r_j) \ dr_j > [1+r_D(\omega_i)]\  \mathcal{K}_j^{(0)} \int_{- \infty}^{+ \infty} w(r_j) \ dr_j,
$$
where
\begin{align*}
& w(r_j)
\\& = \int_{r_1=-\infty}^{+ \infty}  \dots \int_{r_{j-1}} \int_{r_{j+1}} \dots \int_{r_n}  h_{i1}(r_1, \dots, r_j, \dots, r_n) \ f_R(r_1, \dots, r_n) \ dr_n \dots dr_{j+1} \ dr_{j-1} \dots dr_1,
\end{align*}
where
$$ h_{i1}(r_1, \dots, r_j, \dots, r_n) = \frac{\partial (u_i \circ  v)}{\partial P_i^{(1)}}.$$
\item The expected return of the debt issued by Bank $j$ is higher than the cost of funding of Bank $i$ if and only if
\begin{align*}
& \int_{-\infty}^{\frac{L_j^*[1+r_D(\omega_j)]-d_j}{a_j}} (a_j r_j + d_j) w(r_j) \  dr_j+ L_j^{*}[1+r_D(\omega_j)]) \int_{\frac{L_j^*[1+r_D(\omega_j)]-d_j}{a_j}}^{+ \infty} w(r_j) \  dr_j 
\\& > [1+r_D(\omega_i)] \mathcal{L}_j^{(0)} \int_{- \infty}^{+ \infty}  w(r_j) \  dr_j,
\end{align*}
where $d_j=\kappa_j \left( Ax_j^{(0)}+ A\ell_j^{(0)} (1+r_f) \right)$ and $w(r_j)$ has been defined above.
\Eit
    \label{PropPositiveness_2}
\end{prop}

Equation \eqref{SRN} corresponds to the fact that $\mathds{E} \left[ K_j^{(1)} \right]- [1+r_D(\omega_i)]\  \mathcal{K}_j^{(0)} >0 $. Note that in this formula, the return of only Bank $j$ matters. It can be beneficial for Bank $i$ to increase its participation in Bank $j$ if the return on equity of Bank $j$ is higher than the interest rate that Bank $i$ must pay for its debt. 

In the general case, the same type of inequality as \eqref{SRN} is obtained. However, it takes the marginal utility (up to function $v$) into account via $w(r_j)$. For interpretation purpose, let us assume that $v=Id$. For a given value of $r_j$, the algebraic gain of increasing the participation $\pi_{ij}$ must be weighted by the marginal utility, which depends on the returns of all institutions. Integrating this marginal utility with respect to all returns $r_1, \dots, r_n$ apart from $r_j$ yields the term $w(r_j)$. The risk aversion of Bank $i$ is embedded in the term $w(r_j)$.

The same type of argument applies in the case of the debt.


\subsection{Testing of the diversification motive: the network shape}

Let us now compare the consequences of Theorem \ref{Prop_ProgramRNG} and Stylized Facts 2 and 3 on the network shape, and discuss the impact of risk-aversion and limited liability.

A risk-neutral bank with unlimited liability gets interconnected to others by strict mechanical behaviors: it seeks sequentially for the highest returns until binding the solvency constraint. Consequently, the network shape is very structured and directive since everyone gets interconnected in the same direction. Thus, in such a case, there is no general shape.\footnote{Nevertheless, with a particular set of returns, a star network can occur.} In other words, with risk-neutral banks and unlimited liability, the diversification motive cannot provide interesting results.

In the case of risk-averse banks, the interconnections tend to shape a complete network. Institutions carry out a diversification to decrease the variance, in addition to their aim of obtaining higher returns. Note that a diversified portfolio has a lower variance than a concentrated one.\footnote{If $X$ and $Y$ are two random variables with mean $\mu$, variance $\sigma^2$ and correlation $\rho<1$, then \\ $\mathds{E}(X+Y) = 2 \mu = \mathds{E}(2X)$ whereas $\mbox{Var}(X+Y) = 2(1+\rho)\sigma^2 < 4\sigma^2 = \mbox{Var}(2X)$.} Therefore, even if all institutions have similar returns, it can be beneficial to get interconnected. To significantly benefit from the diversification, the variance reduction must be high enough: situations where the specific assets are not almost non-risky and/or where the correlation is negative are prone to yield a complete network structure. These findings will be confirmed numerically in the next section. The limited liability feature can modify the balance between shares and debt securities. 

When considering risk-averse banks, the diversification motive generates complete financial networks, such as those usually observed among major institutions. Therefore, we cannot rule out diversification as explaining interconnections between key financial players.\footnote{Note that our approach has no clue on the relevance of the other motives mentioned in Introduction. We simply show that diversification provides consistent results with empirical observations.}


\section{Network formation and simulation results}
\label{Networkform}
 
In this section, we derive simulation results in order to assess the relevance of the diversification motive for the financial network formation. First, we present the specification that we use and our calibration strategy. Second, we develop a network formation process taking advantage of the strong and tractable theoretical results obtained in the previous section. Then, optimal choices for one financial institution and regarding the whole network are analyzed.

\subsection{Specifications}
For the sake of simplicity, two banks are considered, i.e. $n=2$. Each institution is endowed with a capital amount of $1$, i.e. $K_i^{(0)}=1, i=1,2$. Both institutions have $x \mapsto \log(x)$ as utility function. An initial capital of $1$ implies that the equity value $K_i^{(1)},i=1,2,$ at the optimization horizon is about $1$. Therefore, the objective function is close to be linear over the most likely area, meaning that the banks are only slightly risk-averse.

In order to properly understand the main features of our model, we exclude $A\ell$ and $\omega$ from the control variables. The interest rates paid by the two financial institutions, denoted by $r_{D,1}$ and $r_{D,2}$ are therefore fixed. Moreover, the risk-free interest rate is set to zero: $r_{rf}=0$.


Finally, note that the expectations are computed using Monte-Carlo techniques; $100 \ 000$ simulations ensure a good precision.

\subsection{Calibration strategy}
The gross returns on external assets follow a bivariate log-normal distribution:
\begin{equation}
	\left(
		\begin{array}{c}
			\log\left( \Frac{Ax_1^{(1)} }{ Ax_1^{(0)} } \right) \\
			\log\left( \Frac{Ax_2^{(1)} }{ Ax_2^{(0)} } \right)		
		\end{array}
	\right)
	\sim
	\mathcal{N}
	\left[
		\left(
		\begin{array}{c}
			\mu_1 \\
			\mu_2		
		\end{array}
		\right)
		,
		\left(
		\begin{array}{cc}
			\sigma_1^2 					& \rho_{1,2}\sigma_1\sigma_2  \\
			\rho_{1,2}\sigma_1\sigma_2 	& \sigma_2^2   	
		\end{array}
		\right)
	\right].
\end{equation}
In order to calibrate the mean parameter, we consider the income statement in the Consolidated Financial Statements for BHCs (reporting form FR Y-9C) for banks over \$10 billion. Between 12/31/2010 and 12/31/2012, the (annual) net income varies from 0.51\% to 0.71\% of the total assets. We round this value, considering that on average the net income of our banks is equal to 1\%. Over the same period, the interest expenses represent between 0.74\% and 1.07\% of the total assets.\footnote{In this paper, we consider that the total assets are equal to the earning assets and to the average assets.} We basically consider that the cost of debt ($r_{D,1}$ and $r_{D,2}$) varies between 0\% and 1\%.
Finally, the expected return of the external assets for Bank $i$ is equal to $1\% + r_{D,i} \ell_i$, where $\ell_i$ is the Bank $i$'s ratio of debt over total assets. For the variance parameter, a probability of default of 0.1\% is in line with the current rating of major banks. We combine the informations relative to the net income and the probability of default to compute the parameters $\mu_i$ and $\sigma_i, i=1,2$ (see Appendix \ref{AppendixDynamicCalibration} for details).
The parameter $\rho$ lies between -0.9 to 0.9. A negative $\rho$ can be interpreted as a sign of competition between the two banks or as the fact that banks operate in different markets (or geographical areas). Meanwhile, a positive $\rho$ could be interpreted as an underlying common factor affecting both banks.

We consider the Basel 2 regulation. This regulation does not provide a unique set of values for the risk weights $k_i^A$, $k^{\pi}$ and $k^{\gamma}$. If the external assets correspond to a retail activity (i.e. loans to households), loans to unrated firms (i.e. small firms) or quoted shares, the required capital is equal to 6\%, 8\% or 23.2\% of the total exposure, respectively. For debt securities issued by banks, the required capital is equal to 1.6\% (when AAA or AA rated) or 4\% (when A rated).  Lastly, as discussed in \cite{repullo2013procyclical}, there is a factor between the regulatory capital and the (accounting) equity, that varies from 1 to 2. For the sake of simplicity, we consider that the regulatory capital is either equal to the equity or to a half of the equity. Bottom line, we have 8 possible sets of risk weights.

\subsection{Discussion about the pricing of shares and debt securities}
\label{Chapnetworls_Subsec_Discussion_Pricing}

Recall that the position of Bank 1 at time $t=1$ is as follows if Bank 2 is solvent:
\begin{align}
    \label{Eq_Position_Equilibrium}
    P_1^{(1)} &= Ax_1^{(0)}(1+r_1) + \pi_{12} [ \kappa_2 Ax_2^{(0)}(1+r_2)-L_2^*(1+r_{D,2}) ]+ \gamma_{12} L_2^*(1+r_{D,2}) \nonumber
    \\& - \left(Ax_1^{(0)} + \pi_{12} \mathcal{K}_2^{(0)} +  \gamma_{12} \mathcal{L}_2^{(0)} - K_2^{(0)} \right)(1+r_{D,1}) \nonumber
    \\& = Ax_1^{(0)}(r_1-r_{D,1})+ \pi_{12} \left[ \kappa_2 Ax_2^{(0)}(1+r_2)-L_2^*(1+r_{D,2})-\mathcal{K}_2^{(0)} (1+r_{D,1}) \right] \nonumber
    \\& + \gamma_{12} \left[ L_2^*(1+r_{D,2} ) -\mathcal{L}_2^{(0)}(1+r_{D,1}) \right]+K_2^{(0)} (1+r_{D,1}).
\end{align}
The terms $\mathcal{K}_2^{(0)}$ and $\mathcal{L}_2^{(0)}$ are respectively the market values of the share securities and debt securities issued by Bank 2 at time $t=0$.
In a complete market and with the usual assumptions, the price of an asset would be the discounted expected payoff under the risk-neutral probability:
    $$ \mathcal{K}_2^{(0)}= \frac{\mathds{E}_{RN} [ K_2^{(1)} | \mathcal{F}_0 ]}{1+r_{rf}},$$
where $\mathcal{F}_0$ denotes the available information at time $t=0$.
Since 
$ \forall t \in [0,1], K_2^{(t)} = \max  \left[ \kappa_2 Ax_2^{(t)} - L_2^*(1+r_{D,2}), 0 \right]$, $K_2$ appears as a call option whose underlying is $Ax_2$ and whose strike is $L_2^*(1+r_{D,2})$. However, since $Ax$ is the price of an illiquid asset, it is difficult to argue that there exists a unique probability (the risk-neutral probability) that makes $Ax$ a martingale. Therefore, we choose to consider that the price is the discounted expected payoff under the physical probability. The corresponding prices $\mathcal{K}_2^{(0)}$ and $\mathcal{L}_2^{(0)}$ are given in the following proposition.
\begin{prop}
\label{Prop_PricingMarketValue}
If we assume that $\log\left(Ax_i^{(1)} / Ax_i^{(0)} \right) \sim \mathcal{N}(\mu_i,\sigma_i^2)$, then the expected equity and debt values of Bank $i$ are
    \begin{align*}
        \mathds{E}_0\left(K_i^{(1)} \right)   & = \kappa_i Ax_i^{(0)} e^{\mu_i  + \frac{1}{2}\sigma_i^2 } \left[ 1 - \Phi(\tilde{u}-\sigma_i) \right]-   L_i^{*(1)} \left[ 1 - \Phi(\tilde{u}) \right], \\
        \mathds{E}_0\left(L_i^{(1)}\right)    & = \kappa_i Ax_i^{(0)} e^{\mu_i  + \frac{1}{2}\sigma_i^2 } \Phi\left( \tilde{u}-\sigma_i  \right) +    \left[ 1 - \Phi(\tilde{u}) \right], \\
        \end{align*}
        where $\tilde{u} = \Frac{1}{\sigma_i}\left( \log\left( \Frac{L_i^{*(1)}}{\kappa_i Ax_i^{(0)}} \right) - \mu_i \right)$, $L_i^{*(1)}=L_i^{*(0)} (1+r_{D,i})$ and $\Phi$ is the distribution function of the standard Gaussian variable.
\end{prop}

In order to understand some implications of our pricing choice, consider a situation where all returns are deterministic and $r_2>r_{D,2}$. In such a framework, we have
$$\mathcal{K}_2^{(0)}=\frac{\kappa_2 Ax_2(1+r_2)-L_2^*(1+r_{D,2})}{1+r_{rf}} \mbox{~~and~~} \mathcal{L}_2^{(0)}= \frac{L_2^*(1+r_{D,2})}{1+r_{rf}}.$$
Therefore, injecting these prices in \eqref{Eq_Position_Equilibrium}, we obtain
\begin{align}
    \label{Eq_Position_Physical_Probability}
    P_1^{(1)} &=Ax_1(r_1-r_{D,1})+ \pi_{12} \left[ (\kappa_2 Ax_2(1+r_2)-L_2^*(1+r_{D,2})) \left( 1-\frac{1+r_{D,1}}{1+r_{rf}} \right) \right] \nonumber
    \\& + \gamma_{12} \left[ L_2^*(1+r_{D,2} )  \left( 1-\frac{1+r_{D,1}}{1+r_{rf}} \right) \right]+K_2^{(0)} (1+r_{D,1}).
\end{align}
Generally, we have $r_{D,1}>r_{rf}$, meaning that the factors of $\pi_{12}$ and $\gamma_{12}$ are negative and thus that the net yields on shares and debt securities are negative. Therefore, for a risk-neutral agent (i.e. not interested in variance reduction), it would not be optimal to invest in shares and debt securities. That stems partly from the fact that we have priced these instruments using the physical probability. Under the latter probability, the shares and debt securities yield in average the risk-free rate. This feature could of course be challenged. Note that we should pay attention to the interpretations based on \eqref{Eq_Position_Physical_Probability} since \eqref{Eq_Position_Physical_Probability} only gives the expression of the position in a very simplified case. Equation \eqref{Eq_Position_Physical_Probability} must only be considered as an indication.

Contrary to the share and debt security prices, the initial value of $Ax_1$ does not take the future returns into account. As we already mentioned, $Ax_1$ is an illiquid asset that cannot be exchanged on the market. Therefore, the assumption of absence of arbitrage is not necessarily satisfied and we price $Ax_1$ using its book value. Since generally $r_1 > r_{D,1}$, the specific asset $Ax_1$ provides a positive return. This is logical since getting positive returns via maturity transformation constitutes the core business of banks. However, in the pricing of $\mathcal{K}_2^{(0)}$, we consider the future returns of $Ax_2$. This asymmetry can be discussed but it is difficult to find an ideal solution given the close link between a market asset ($\mathcal{K}_2^{(0)}$) and an illiquid asset ($Ax_2$) in our model.

\subsection{Methodology for the network formation}

The optimization programs $\mathcal{P}_i$ and $\mathcal{P}_i^{'}$ presented in Section \ref{TheoreticalProp} allow computing the balance sheet of an institution, knowing the state of the others. Here the aim is to build a complete network using this individual optimization program. To this purpose, we operate in a sequential way until an equilibrium in the network is reached.

We propose to use an iterative game. At each step, one institution optimizes its balance sheet taking into account the state of the network obtained at the previous step.
Thanks to Corollary \ref{PropChoix}, there exists only one network at each step. The procedure is as follows\footnote{Note that this formation process can be applied in the general framework of Section \ref{TheoreticalProp} but is here presented using the previously mentioned specification.}:
\begin{enumerate}
    \item Bank $1$ optimizes its balance sheet on $Ax_1$ and $L_1^*$. Quantities $\pi_{1,2}$ and $\gamma_{1,2}$ are forced to be equal to zero since at the initialization step, Bank $2$'s balance sheet is totally unknown;
    \item Bank $2$ optimizes its balance sheet on $Ax_2$, $L_2^*$, $\pi_{2,1}$ and $\gamma_{2,1}$ given Bank $1$'s balance sheet from step 1;
    \item Bank $1$ optimizes its balance sheet on $Ax_1$ ,$L_1^*$, $\pi_{1,2}$ and $\gamma_{1,2}$ given Bank $2$'s balance sheet from step 2. $\pi_{1,2}$ and $\gamma_{1,2}$ are optimized for the first time;
    \item Bank $2$ optimizes its balance sheet on $Ax_2$, $L_2^*$, $\pi_{2,1}$ and $\gamma_{2,1}$ given Bank $1$'s balance sheet from step 3;
    \item Bank $1$ optimizes its balance sheet on $Ax_1$ ,$L_1^*$, $\pi_{1,2}$ and $\gamma_{1,2}$ given Bank $2$'s balance sheet from the previous step;
    \item and so on.
\end{enumerate}
For further details, see Appendix \ref{AppendixAlgoNetworkFormation}.

Theoretically, this procedure may be endless. However, in less than 10 steps, the variations of the control variables from one step to the next are lower than 1\% and we consider that the final situation constitutes an equilibrium. Moreover, if we accept the numerical argument for the existence of the limit-network, we can affirm its uniqueness. Indeed, if at each step the network is unique, then its final state is necessarily unique. It is interesting to note that this method is inspired by the classical methodology used to determine a Nash equilibrium (in the sense that no institution has any interest in deviating from its current state). However, further investigations would be required to know if the network obtained by our method effectively corresponds to a Nash equilibrium.

Last but not least, it is important to check that the obtained network is consistent in the sense that it satisfies \eqref{K} and \eqref{Lx}. Firstly, at time $t=0$, all banks considered in the network are solvent; otherwise they would disappear from the network. That means that the initial debt equals the contractual one: $L_i^{(0)}=L_i^*, i=1, \dots, n$. Therefore, \eqref{Lx} is automatically satisfied for each institution. Moreover, at each step, being a constraint of the optimization program, \eqref{K} is satisfied for the bank optimizing its balance sheet. If preliminary, this step has impacts on the other banks' balance sheets and \eqref{K} is not exactly satisfied anymore for them. Nevertheless, after some iterations, the network does not evolve from one step to the next (due to the convergence), implying that \eqref{K} is satisfied for all institutions. These two points show that the obtained network is actually consistent.

\medskip
\noindent This sequential algorithm could appear a little artificial but it is actually close to what happens in reality. An example of a real formation process of a network is as follows:
\Ben
    \item Consider an initial situation where there is no bank;
    \item A first bank, denoted by $B_1$, is created during year $t=0$. Since there are no other banks, there are no possible interconnections. Thus, $B_1$ optimizes $Ax_1$ and $L_1$. On January 1st of year $t=1$, $B_1$ publishes its balance sheet;
    \item Imagine that on January 3rd, a second bank $B_2$ is created. $B_2$ knows $Ax_1$ and $L_1$ and then can solve the optimization program to determine $Ax_2$,  $L_2$, $\pi_{2,1}$ and $\gamma_{21}$. Once proportions $\pi_{2,1}$ and $\gamma_{2,1}$ have been determined, $B_2$ can buy on the secondary market shares and bonds issued by $B_1$ in these proportions;
    \item On June 1st, $B_1$ and $B_2$ publish their balance sheets (apart from interconnections). Since the balance sheet of $B_1$ did not evolve since January 1st, $B_2$ has no new optimization to carry out. On the other hand, $B_1$ discovers for the first time informations relative to $B_2$: $Ax_2$ and $L_2$. Then $B_1$ optimizes its balance sheet and thus obtains $Ax_1$, $L_1$, $\pi_{12}$ and $\gamma_{12}$. $B_1$ can buy on the secondary market shares and bonds issued by $B_2$;
    \item On January 1st of year $t=2$, balance sheets of $B_1$ and $B_2$ are published. The balance sheet of $B_2$ did not change and thus $B_1$ has no optimization to do. On the other hand, $B_2$ must adapt to the new balance sheet of $B_1$;
    \item and so on.
\Een
After such iterations, one may think that there is convergence to an equilibrium in the network. Balance sheets of $B_1$ and $B_2$ do not evolve a lot from one step to the next.

\subsection{Simulation results about the optimal choice for one institution}
Let us here focus on the second step of the iterative game where Bank 2 optimizes its whole balance sheet (knowing the choice of Bank 1 at step 1). We assume that Bank $1$'s external assets are equal to 10. We present the sensitivity of the optimal choices of external assets $Ax_2$, nominal debt $L_2^*$ and interconnections $\pi_{2,1}$ and $\gamma_{2,1}$, with respect to the regulatory parameters and correlation $\rho$. Our computations were carried out under various debt-issuing conditions (not costly with $r_{D,1}=r_{D,2}=r_{rf}=0$, both costly with $r_{D,1}=r_{D,2}=1\% > r_{rf}=0$ and only one costly with $r_{D,1}=1\% > r_{D,2}=r_{rf}=0$) and we observe that the results are independent of these conditions. In each set-up, we consider the 8 sets of risk-weights and we let the correlation parameter vary between $-0.9$ and $+0.9$.

The corresponding results are summarized in Table \ref{TabStylizedRestultsIndividual}. First, we observe that interconnections based on debt securities are never used. A direct consequence is that the risk weight on debt, $k_{\gamma}$, has no impact on the balance sheet and thus does not appear in Table \ref{TabStylizedRestultsIndividual}. Second, interconnections based on share securities are used only when the correlation is lower than -0.3 (independently of the interest rates) and when the associated risk weight is equal to 23.2\%. They linearly decrease from about 45\% to 0\% between $\rho=-0.9$ and $\rho=-0.3$. Third, the solvency constraint is binding. The optimal external assets represent about $1/k^A$. The last row-block displays the ratio of interbank assets over the total assets: when interconnections are present, their proportion in the total assets is in line with the stylized facts.

These results could be interpreted as follows. First, the bank plays its core business: it invests as much as it can in its external assets. Then, if the regulation is not too strict and if the competitor's results are sufficiently anti-correlated, the bank opts for diversification: it slightly lowers its external assets to buy share securities issued by the competitor. Debt securities are not used since their net returns are negative (as a consequence of the pricing specification described in Section \ref{Chapnetworls_Subsec_Discussion_Pricing}) and "nearly" deterministic (due to the low probability of default).

\begin{table}[!h]
    \centering
    \begin{tabular}{cccccc}
             & $k^{\pi}$    & $k^A$     & $\rho=-0.9$ & $\rho=-0.6$ & $\rho=-0.3$   \\ \hline
        $Ax$ & 23.2\%       & 6\%       & 14          &   15        &   16          \\
             & 23.2\%       & 8\%       & 11          &   12        &   11          \\
             & 46.4\%       & 12\%      & 8           &   8         &    8          \\
             & 46.4\%       & 16\%      & 6           &   6         &    6          \\ \hline
       $\pi$ & 23.2\%       & 6\%       & 45          &   25        &   0          \\
       (\%)  & 23.2\%       & 8\%       & 45          &   25        &   0          \\
             & 46.4\%       & 12\%      & 0           &   0         &    0          \\
             & 46.4\%       & 16\%      & 0           &   0         &    0          \\ \hline
    $\gamma$ & 23.2\%       & 6\%       & 0           &   0         &   0          \\
       (\%)  & 23.2\%       & 8\%       & 0           &   0         &   0          \\
             & 46.4\%       & 12\%      & 0           &   0         &    0          \\
             & 46.4\%       & 16\%      & 0           &   0         &    0          \\ \hline
    $IBA/TA$ & 23.2\%       & 6\%      & 3.1	      & 1.6         &	0       \\
	   (\%)  & 23.2\%       & 8\%      & 3.9	      & 2.0         &	0          \\
	         & 46.4\%       & 12\%      & 0           &	0           & 	0       \\
	         & 46.4\%       & 16\%      & 0	          & 0           & 	0 
    \end{tabular}
    \caption{Stylized results for the optimal choice of one institution, when $r_{D,1}=r_{D,2}=0.$}
    \label{TabStylizedRestultsIndividual}
\end{table}

\subsection{Iterative game results}
The iterative game reaches an equilibrium in less than 5 steps. The features pictured in the analysis of the behavior of one institution are still present. Especially, results are robust to the debt-issuing conditions.

Both institutions have the same balance sheet, whose composition is given in Table \ref{TabStylizedRestultsIterativeRiskFreeRate}. Results are very similar to those for one institution only (Table \ref{TabStylizedRestultsIndividual}). In particular, the proportion of interbank assets in the total assets is in agreement with the sylized facts. 
Note that for $\rho=-0.9$ and $\rho=-0.6$, the values of $\gamma_{12}$ and $\gamma_{21}$ are close to $10^{-4}$. However, we have reported $0$ since such low values do not have any economic meaning.


Let us state that these results have been obtained using $\kappa_i=1, i=1,2$, in order to avoid numerical instability. Indeed, if the values of $\kappa_i$ become too large, it makes no sense anymore to assume that the asset side of the other banks is only composed of their external assets.

\begin{table}[!h]
    \centering
    \begin{tabular}{cccccc}
             & $k^{\pi}$    & $k^A$     & $\rho=-0.9$ & $\rho=-0.6$ & $\rho=-0.3$   \\ \hline
        $Ax$ & 23.2\%       & 6\%       & 15          &   15       &   16          \\
             & 23.2\%       & 8\%       & 11          &   12      &   12          \\
             & 46.4\%       & 12\%      & 8           &   8         &    8          \\
             & 46.4\%       & 16\%      & 6           &   6         &    6          \\ \hline
       $\pi$ & 23.2\%       & 6\%       & 70        &     45     &   16          \\
       (\%)  & 23.2\%       & 8\%       & 60        &     35     &   6         \\
             & 46.4\%       & 12\%      & 0           &   0         &    0          \\
             & 46.4\%       & 16\%      & 0           &   0         &    0          \\ \hline
    $\gamma$ & 23.2\%       & 6\%       & 0           &   0        &   0          \\
       (\%)  & 23.2\%       & 8\%       & 0           &   0        &   0          \\
             & 46.4\%       & 12\%      & 0           &   0         &    0          \\
             & 46.4\%       & 16\%      & 0           &   0         &    0          \\ \hline
    $IBA/TA$ & 23.2\%       & 6\%       & 3.2         &	2.3         & 1 \\
        (\%) & 23.2\%       & 8\%       & 3.6         &	2.4        & 0.5 \\ 
             & 46.4\%       & 12\%      & 0           &   0         &    0          \\ 
             & 46.4\%       & 16\%      & 0           &   0         &    0          \\ 
    \end{tabular}
    \caption{Stylized results for the iterative game, when $r_{D,1}=r_{D,2}=0$.}
    \label{TabStylizedRestultsIterativeRiskFreeRate}
\end{table}

\subsection{Testing the diversification motive}

Regarding the capacity of the diversification motive to account for interconnections, the previous results provide a quantitative assessment completing the qualitative arguments developed in Section 3. The key result is that when returns on specific assets are anti-correlated, diversification leads to interconnections with reasonable size in terms of proportion of the total assets. However, debt securities are never used, meaning that interconnections are only supported by share securities. This portfolio composition contrasts with empirical findings. 

Nevertheless, it is important to emphasize that in our simulation study, the choice of pricing shares and debt securities under the physical probability has large impacts. As explained in Section \ref{Chapnetworls_Subsec_Discussion_Pricing}, it implies that the net yields of shares and bonds are negative. Therefore, in this framework, interconnections only allow for variance reduction but not for gain opportunity. We can expect this feature to be modified if the pricing is done under the risk-neutral probability. Interconnections in both shares and debt securities could then be observed, even for values of $\rho$ larger than $-0.3$. The study of the risk-neutral specification constitutes an ongoing work. In some sense, these two types of specification for the pricing allow disentangling the two aims of the diversification: variance reduction and opportunity.

The latter discussion shows that our model seems promising but that results are very sensitive to the different possible specifications. Moreover, one feature that is not included in our model for the sake of simplicity may partly explain the discrepancy regarding debt securities. In reality, there are additional constraints -apart from the required capital- imposed to large shareholders, such as mandatory public communication. These constraints could discourage banks to invest in shares and could instead lead to higher investments in debt securities. 

\section{Application: impact of interconnectedness regulation}
\label{SecApplications}

The diversification motive has proven an interesting explanation of the bank size (Stylized Fact 1), the network shape (Stylized Facts 2 and 3) and the composition of interconnections (Stylized Fact 4). Previous results concern the initial network resulting from banks' choices based on their expectations.
Due to the endogenous feature of interconnections, we can build some plausible counterfactual scenarios, allowing to analyze the impact of regulation on the welfare at time $t=1$. 

\subsection{Assessing interconnections}
The interconnectedness across financial institutions has become a key concern of supervisors and regulatory authorities. Currently, long-term interbank exposures are covered by two main requirements. The first one concerns the solvency required capital for the interconnections, as for any other assets. It imposes a constraint on the total interbank exposure. The second one concerns "large" single exposures and imposes the risk-weighted exposure to be lower than a fraction of the equity.\footnote{We do not distinguish equity, own funds and regulatory capital.} Currently, the Basel Committee considers that an exposure is large if above 5\% (instead of 10\%) of the equity and to impose that the risk-weighted exposure ($k^{\pi} \pi_{ij} K_j + k^{\gamma} \gamma_{ij} L_j$ for the exposition to Bank $j$) has to be lower than 25\% of the equity \cite[see][Section II and Section IV.B]{BCBS2014}. These requirements are valid for any type of exposure (e.g. corporate or sovereign) but the weights can vary with respect to the type. Moreover, the Basel Committee proposes to introduce tighter rules about interbank exposures for the G-SIBs (Global Systematically Important Banks). An upper bound between 10\% and 15\% instead of 25 \% is in discussion \citep[see][Section V]{BCBS2014}. These tighter rules about interbank exposures aim at reducing the risk of contagion.

These different aspects show that interconnectedness is generally assessed in a negative way. Actually, supervisors are primarily concerned with excessive risks and therefore either analyze the effects of interconnections under depressed scenarios (stress-test approach) or build indicators in order to monitor the current fragility of the financial sectors. In both approaches, interconnectedness usually means contagion only. For instance, the seminal papers about network stress-tests -such as \cite{furfine2003interbank} on US data or \cite{upper2004estimating} on German data- sequentially consider the effects in their national banking sectors of the default of each bank. From their point of view, interconnected banks are likely to trigger defaults or to go bankrupt due to contagion.

Nevertheless, these analyses are not built on counterfactuals. They certainly give informative insights about what could happen within the current network in the case of defaults of some institutions or difficult macroeconomic conditions. However, since the network reaction is not taken into account, such studies do not really provide any clue on the way to obtain a more resilient network structure. Moreover, note that the question of regulation impact has hardly been addressed quantitatively, even in the case of a crystallized network.

The endogenous nature of interconnections in our model precisely allows us to study how the network reacts to tough macroeconomic conditions or to assess the impact of regulation on interbank exposures, for instance of the regulation in discussion at the Basel Committee. In the following, we focus on the impact of regulatory changes. To do so, we consider our 8 sets of regulatory weights associated to interbank exposure ($k^A$, $k^{\pi}$ and $k^{\gamma}$).\footnote{In reality only 4 since with the specification chosen, $k^{\gamma}$ has no impact.} For each specific set, the initial network is derived using our formation process. This step accounts for the diversification motive. Then we simulate returns of the external assets and examine the network at time $t=1$. Let us emphasize that the shocks are properly propagated through the real interconnections.\footnote{Contrary to the assumption -used in the individual optimization program- that banks do not consider interconnections of their counterparts.} The unique set of values $K_i$ and $L_i$ (see Proposition \ref{Prop_Prop2_Gouretal}) is determined using the algorithm described in Appendix \ref{AppendixAlgoWelfareComputation}. This allows us to carry out a fair assessment of contagion. To do so, we build a welfare indicator including an explicit concern for the real economy and examine its sensitivity to the regulatory set of weights. 

\subsection{Welfare analysis}
We adapt the welfare analysis by \cite{repullo2013procyclical} to assess the impact of the regulatory parameters on the real economy.

The contribution of one bank is either negative or positive. When a bank defaults, its contribution is negative and proportional to the loss on its debt. This feature encompasses the cost of deposit insurance. When a bank is solvent, its contribution is the volume of external assets, i.e. the lendings provided to the real economy. This component captures the capacity to finance the real economy. The contribution of Bank $i$ is written
$$
    w_i = - c \  \left(L_i^{*(1)} -  L_i^{(1)} \right)           + Ax_i^{(1)},
$$
where
$c$ is the social cost for deposit insurance (in \cite{repullo2013procyclical}, $c$ varies in $[0, 60\%]$).

Our welfare indicator $W$ is the ratio of the contribution of all banks over the initial lending to the real economy:
$$    W = \frac{ w_1+w_2} {Ax_1^{(0)}+Ax_2^{(0)}}= \frac{ Ax_1^{(1)}+Ax_2^{(1)} -c \left(L_1^{*(1)} -  L_1^{(1)} + L_2^{*(1)} -  L_2^{(1)} \right) }{Ax_1^{(0)}+Ax_2^{(0)}}.
$$
For $c=0$, the welfare is given in Table \ref{Chapnetworks_Fig_Welfare}. When there are interconnections, the welfare is higher than 1, indicating an increase of the banking capacity to lend to the real economy. In contrast, when there is no interconnection, the value of the external assets decreases. A complete analysis of the impact of interconnections would require further studies. However, these results suggest that the interconnections stemming from diversification are beneficial for the real economy.
\begin{table}[!h]
    \centering
    \begin{tabular}{cccccc}
                                & $k^{\pi}$    & $k^A$     & $\rho=-0.9$ & $\rho=-0.6$ & $\rho=-0.3$    \\ \hline
        Sum of contributions    & 23.2\%       & 6\%       & 29.9        &   30.9     &  32.4           \\
                                & 23.2\%       & 8\%       & 22.8        &   23.6     &  24.8           \\
                                & 46.4\%       & 12\%      & 15.6        &   15.6     &  15.6           \\
                                & 46.4\%       & 16\%      & 11.9        &   11.9     &  11.9           \\ \hline
        Welfare (\%)            & 23.2\%       & 6\%       & 101.0       &  101.0     & 101.0           \\
                                & 23.2\%       & 8\%       & 101.0       &  101.0     & 100.8           \\
                                & 46.4\%       & 12\%      & 93.4        &  93.4      & 93.4            \\
                                & 46.4\%       & 16\%      & 95.6        &  95.6      & 95.5
    \end{tabular}
    \caption{Welfare.}
    \label{Chapnetworks_Fig_Welfare}
\end{table}


\section{Concluding remarks}

A diversification motive appears as a sound candidate to account for long-term exposures across financial institutions. The first aim of this paper is to test this assumption.

To this purpose, we build a model of financial network in which the balance sheets of all institutions (including interconnections) are totally endogenous apart from the equity.
The network formation process involves two components. The first one explains how a bank optimizes its balance sheet knowing the state of the other banks in the network. We prove the existence and partial uniqueness of the solution of this optimization. The second part shows how to form the network using the individual optimization program. The existence and unicity of this network are shown by numerical arguments.
An important feature of our model is its ability to account for the main features of the banking and the insurance business with the same set of parameters. Nevertheless, we focus in this paper on the banking business.

Secondly, the characteristics of the resulting network are compared to features usually observed.
As to the shape of the network, we theoretically find that the diversification motive leads to a network close to those observed across big banks. Regarding the size and support of the interconnections, we show that a correct magnitude is reached under standard calibration. Moreover, the results are sensitive to some specifications, for example the pricing method of shares and debt securities.

The fact that our network is totally endogenous allows studying how it adapts to regulatory changes. Thus, the second aim is to apply our model to fairly assess the impact of regulation on interbank exposures. To this purpose we study the evolution of the welfare with respect to the regulatory weights $k^A$ and $k^{\pi}$. We observe that the welfare is higher under regulations favoring interconnections.

Ongoing work includes the complete study in the case of insurance companies and the extension to short-term interconnections. An exhaustive sensitivity analysis of the obtained network with respect to macroeconomic parameters like the returns's means as well as other specifications -e.g. concerning the pricing of shares and debt securities- are also under study. Finally, a simulation exercise in the case of 3 or 4 banks would also be of great interest.

\newpage
\section*{Acknowledgements}
We are very grateful to Paul Embrechts for his very detailed and useful comments and suggestions. We also thank Christian Gouri\'eroux, Claire Labonne, Christian Y. Robert and the participants of the French Financial Association (AFFI) Summer Conference 2013 in Lyon, the 3${}^{rd}$ International Conference of the Financial Engineering and Banking Society 2013 in Paris, the 5${}^{th}$ International Conference of International Finance and Banking Association (IFABS) 2013 in Nottingham, the 12${}^{th}$ International Conference on Credit Risk Evaluation (CREDIT) 2013 in Venise, the Workshop Trade and Network 2013 in Leuven, the BFMI Conference 2013 in Surrey, the 7${}^{th}$ International Conference on Computational and Financial Econometrics (CFE) 2013 in London and the 7${}^{th}$ Risk Forum 2014 in Paris, for their useful comments and suggestions. Erwan Koch has been partially financially supported by the project MIRACCLE-GICC. He also would like to thank RiskLab at ETH Zurich and the Swiss Finance Institute for financial support.

\newpage
\appendix
\section{Example of public information on banks' balance sheets}
\label{AppendixExcerptBHC}
\begin{figure}[!h]
    \centering
    \includegraphics[scale=0.5]{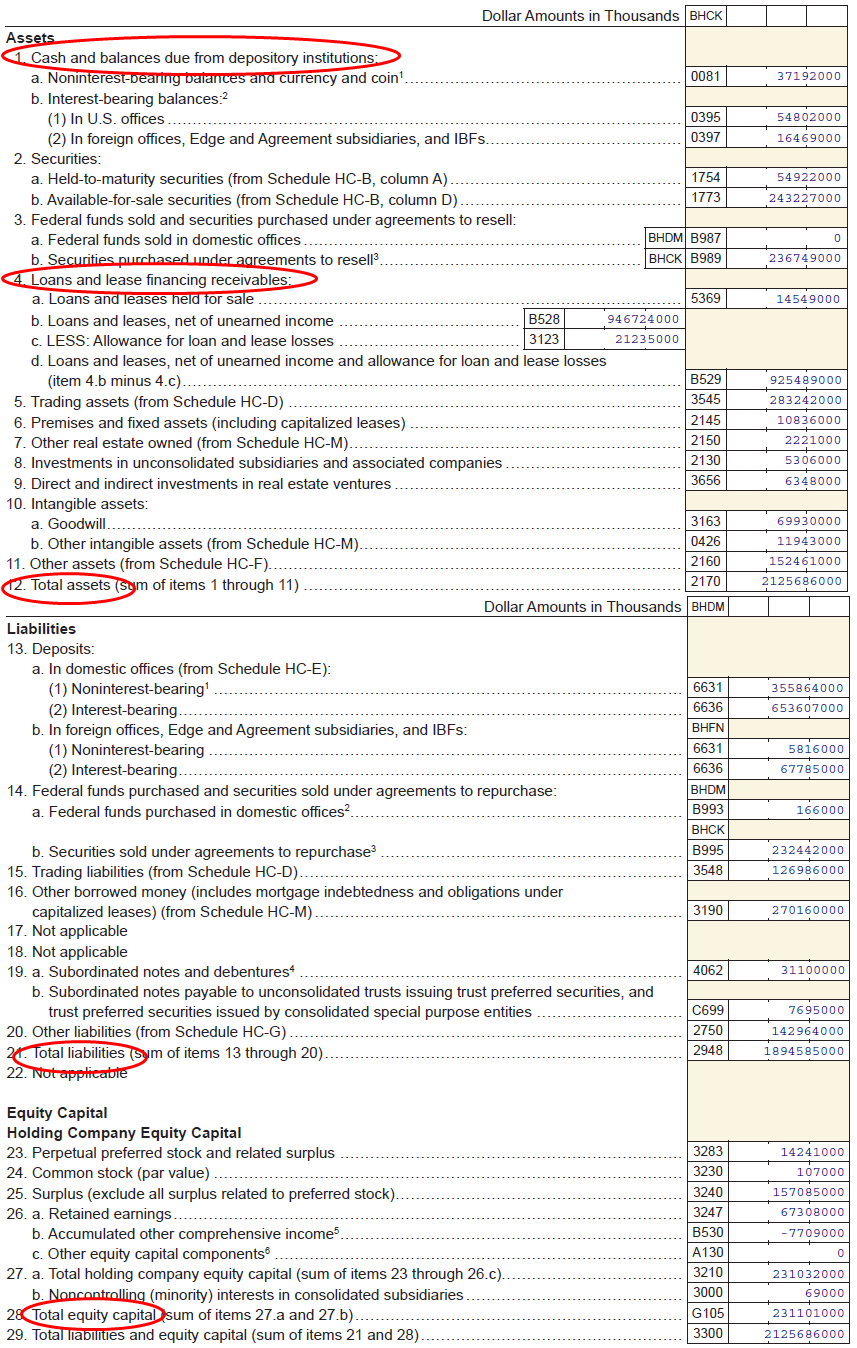}
    \caption{Excerpt of the Consolidated Financial Statements for BHCs of Bank of America at 06/30/2013. Source: www.ffiec.gov.}
    \label{FigExBSBoA}
\end{figure}

\section{The model of \cite{gourieroux2012bilateral}}

In this part, we expose the model of \cite{gourieroux2012bilateral}, that provides the conditions defining an equilibrium between $n$ financial institutions intertwined through shares and debt securities.

\subsection{Existence and uniqueness of the equilibrium}
\label{Appendix_Prop2_Gouretal}

\begin{prop}
\label{Prop_Prop2_Gouretal}
Let us denote by $\Mb{K}=(K_i)_{i =1, \dots, n}$, $\Mb{L}=(L_i)_{i =1, \dots, n}$, $\Mb{L^*}=(L_i^*)_{i =1, \dots, n}$, $\Mb{Ax}=(Ax_i)_{i =1, \dots, n}$ and $\boldsymbol{A\ell}=(A\ell_i)_{i =1, \dots, n}$.
There exists a unique liquidation equilibrium, that is a unique set of values for $\Mb{K}$ and $\Mb{L}$ for any given values of $\Mb{L^*}$, $\Mb{Ax}$, $\boldsymbol{A\ell}$ if for all $i,j=1, \dots, n$:
\begin{itemize}
    \item[$\bullet$] $(A1')$  we have $\pi_{i,j} \geq 0$, $\gamma_{i,j} \geq 0$;
    \item[$\bullet$] $(A2')$  we have $Ax_i \geq 0$, $A\ell_i \geq 0$, $L_i^* \geq 0$;
    \item[$\bullet$] $(A3')$  we have $\Sum^n_{i=1} \pi_{i,j} < 1 $, $\Sum^n_{i=1} \gamma_{i,j} < 1 $.
\end{itemize}
\end{prop}

\begin{proof}
See \cite{gourieroux2012bilateral}.
\end{proof}

Assumptions $(A1')$ and $(A2')$ define a proper space for the parameters: all elements composing the balance sheet must obviously be non-negative.

Assumption $(A3')$ means that some shareholders and creditors do not belong to the perimeter of the selected financial institutions. In practice, the first part of $(A3')$ is generally satisfied providing that we consider consolidated groups. Indeed, the empirical evidence in the studies by \cite{gauthier2012macroprudential} and \cite{Alves2013} clearly shows that $\sum^n_{i=1} \pi_{i,j} < 1 $.
The constraint on the $\gamma_{ij}$ is largely satisfied since core deposits (deposits from external agents) represent approximately $55\%$ of a bank's debt.

\subsection{Case of two financial institutions}
For illustrative purposes, let us consider a network of two institutions whose balance sheets are shown in Table \ref{TabBS12}. In such a case the equilibrium equations (\ref{K})-(\ref{Lx}) are
\begin{equation}
    \left\{
    \begin{array}{rl}
        K_1 = & \max \Big( \pi_{1,1}K_1 + \pi_{1,2} K_2 + \gamma_{1,2} L_2 + A\ell_1 + Ax_1 - L_1^*, 0 \Big), \\
        L_1 = & \min\Big( \pi_{1,1}K_1 + \pi_{1,2} K_2 + \gamma_{1,2} L_2 + A\ell_1 + Ax_1 , L_1^* \Big), \\
        K_2 = & \max \Big( \pi_{2,1}K_1 + \pi_{2,2} K_2 + \gamma_{2,1} L_1 + A\ell_2 + Ax_2 - L_2^*, 0 \Big), \\
        L_2 = & \min\Big( \pi_{2,1}K_1 + \pi_{2,2} K_2 + \gamma_{2,1} L_1 + A\ell_2 + Ax_2 , L_2^* \Big). \\
    \end{array}
    \right.
\end{equation}
One can identify $4$ regimes depending on the situations of Institutions 1 and 2, respectively. These regimes, represented in Figure \ref{FigRegimes}, are:
\begin{itemize}
	\item Regime 1: both Institutions 1 and 2 are solvent;
	\item Regime 2: both Institutions 1 and 2 default;
	\item Regime 3: Institution 1 defaults while Institution 2 is solvent;
	\item Regime 4: Institution 1 is solvent while Institution 2 defaults.
\end{itemize}

\begin{table}[H]
    \begin{center}
        \begin{tabular}{c|ccc|c}
            \multicolumn{2}{c}{Institution 1} & & \multicolumn{2}{c}{Institution 2} \\
            Asset               & Liability        & ~~~~~~~ & Asset  & Liability        \\ \cline{1-2}\cline{4-5}
            $\pi_{1,1} K_1$     & $L_1$            &         & $\pi_{2,1} K_1$     & $L_2$ \\
            $\pi_{1,2} K_2$     & $K_1$                &         & $\pi_{2,2} K_2$ & $K_2$ \\
            $\gamma_{1,2} L_2$  &                  &         & $\gamma_{2,1} L_1$  \\
            $A\ell_1$           &                  &         & $A\ell_2$ &\\
            $Ax_1$              &                  &         & $Ax_2$    & \\
        \end{tabular}
    \end{center}
    \caption{Balance sheets of Institutions 1 and 2.}\label{TabBS12}
\end{table}

\begin{figure}[H]
    \begin{center}
         \begin{tikzpicture}[scale=0.75]
        \fill[color=green!20] (4,2) -- (9,0) -- (9,0) -- (10,0) -- (10,8) -- (0,8) -- (0,7) ;

        \fill[color=blue!20] (4,2) -- (4.5,0) -- (9,0)  ;
        \fill[color=blue!40] (4,2) -- (0,2.5) -- (0,7)  ;
        \fill[color=red!20] (0,0) -- (4.5,0) -- (4,2) -- (0,2.5) ;

        \draw[line width = 1pt,->] (0,0) -- (0,8) node[anchor=east]{$Ax_2+A\ell_2$};
        \draw[line width = 1pt,->] (0,0) -- (10,0) node[anchor=north]{$Ax_1+A\ell_1$};
        \draw[dashed] (4,0) node[anchor=north]{$Ax_1^*$} -- (4,2) ;
        \draw[dashed] (0,2) node[anchor=east]{$Ax_2^*$}-- (4,2) ;
        \draw[line width = 0.5pt,-] (4,2) -- (4.5,0) ;
        \draw[line width = 0.5pt] (4,2) -- (9,0)  ;
        \draw[line width = 0.5pt] (4,2) -- (0,2.5) ;
        \draw[line width = .5pt] (4,2) -- (0,7);
        \node (Regime1et2Defaut) at (2,1) {$\mathcal{R}2$} ;
        \node [text width=2cm,text centered] (Regime1Defaut) at (1.1,4) {$\mathcal{R}3$} ;
        \node [text width=2cm,text centered] (Regime2Defaut) at (5.75,0.5)  {$\mathcal{R}4$} ;
        \node (RegimeSansDefaut) at (6,4) {$\mathcal{R}1$} ;
        \end{tikzpicture}
    \end{center}
    \caption{Regimes.}\label{FigRegimes}
\end{figure}

Figure \ref{FigRegimes} motivates the existence of interconnections between institutions. In a situation without interconnections, the 4 regimes would be defined by rectangles. Here the bounds deviate due to the presence of interconnections. In the case where the external assets of Institution 2, $Ax_2 + A\ell_2$, are just above the limit value $Ax_2^*$, if it is interconnected and if $Ax_1+A\ell_1$ is low, then Institution 2 can default ($\mathcal{R}_2$ is larger in the presence of interconnections). In this case, interconnections have a  negative effect since the predicament of Institution 1 negatively impacts Institution 2 by contagion. When $Ax_2+A\ell_2$ is very low, Institution 2 necessarily defaults if not linked to Institution 1. However, if Institution 2 owns shares of Institution 1, Institution 2 can survive if the external assets' value of Institution 1 is sufficient ($\mathcal{R}_1$ is larger in the presence of interconnections). In such a case, Institution 2 takes advantage of the high yield investments of Institution 1. Thus, we understand that the impacts of interconnections are not necessarily negative and must be fairly assessed.

\section{Proofs}

\subsection{For Theorem \ref{PropExistenceSolOptiProSol}}

\begin{proof}
For $i=1, \dots, n$, let us denote the vectors of all control variables  by $\Mb{X}$. We have
$$    \mathbf{X}= \left( Ax_i^{(0)}, A\ell_i^{(0)}, L_i^{(0)}, \omega_i, \pi_{i,1}, \dots, \pi_{i,n}, \gamma_{i,1}, \dots, \gamma_{i,n} \right)^{'} \in \mathcal{X}_{ad},
$$
where $\mathcal{X}_{ad}$ is the admissible space satisfying all constraints of $\mathcal{P}_i$. From now on, for the sake of notational simplicity, we omit the dependence in $i$ of the vectors containing the control variables.

The proof relies on the Weierstrass Theorem: a continuous function on a compact set reaches its bounds. Therefore, we first show the continuity of the objective function and then the compactness of the admissible set $\mathcal{X}_{ad}$.

\subsection*{Continuity of the objective function}

Under $(A2)$ and $(A3)$, both $u_i$ and $F_R$ are continuous. Therefore, the expectation is also continuous and the objective function is continuous. 

\subsection*{Compactness of the admissible set $\mathcal{X}_{ad}$}
To prove the compactness of $\mathcal{X}_{ad}$, we show that it is a closed and a bounded set. Before, we prove that $\mathcal{X}_{ad}$ is not empty.

\subsubsection*{$\mathcal{X}_{ad}$ is non-empty:}
Let us consider the vector of parameters $\mathbf{X}_0$ defined as
$$    \mathbf{X}_0 = \left( K_i^{(0)}-k^Ll(0,0), k^Ll(0,0), 0, \dots, 0 \right)^{'}.
$$
All constraints apart from $Ax_i \geq 0$, $(BC)$, $(SC)$ and $(LC)$ are obviously satisfied. The inequality $Ax_i \geq 0$ imposes that $K_i^{0} \geq k^L l(0,0)$ which is not restrictive due to the low value of $k^L$ and the fact that $l(0,0)$ can be taken equal to one. The constraint $(BC)$ reduces to $K_i^{(0)}- k^L l(0,0) + k^L l(0,0) = K_i^{(0)}$ and is thus satisfied. 
The constraint $(SC)$ is written
\begin{align*}
    K_i^{(0)} \geq k_i^A Ax_i^{(0)}   & \Longleftrightarrow K_i^{(0)} \geq k_i^A [ K_i^{(0)}- k^L l(0,0)] \\
                                    & \Longleftrightarrow K_i^{(0)}(1-k_i^A) \geq -k_i^A \ k^L l(0,0).
\end{align*}
Due to the inequality $k_i^A<1$ and the positivity of $k_i^A$, $k^L$ and function $l$, the left hand term is positive whereas the right one is negative, giving that $(SC)$ is satisfied.\\
Thus $\mathbf{X}_0 $ belongs to the admissible set $\mathcal{X}_{ad}$, which is therefore not empty.

\subsubsection*{$\mathcal{X}_{ad}$ is a closed set:}
In order to show that the admissible set $\mathcal{X}_{ad}$ is a closed set, we show that it is the intersection of closed sets. 

i) The constraint $(BC)$ can be written
$$    Ax_i^{(0)} + A\ell_i^{(0)} + \Sum_{j=1}^n \pi_{i,j} \mathcal{K}_j^{(0)} + \Sum_{j=1}^n \gamma_{i,j} \mathcal{L}_j^{(0)} - K_i^{(0)} - L_i^{(0)} = 0.
$$
The corresponding admissible space is the reciprocal image of the singleton $\{ 0 \}$, which is a closed set of $\mathds{R}$, by a continuous function. Therefore, $(BC)$ defines a closed set. 

ii) The constraint $(SC)$ is derived in
$$
    k_i^{A} \ Ax_i^{(0)} + k^{\pi} \Sum_{j=1}^n \pi_{i,j} \mathcal{K}_j^{(0)} + k^{\gamma} \Sum_{j=1}^n \gamma_{i,j} \mathcal{L}_j^{(0)} - K_i^{(0)} \leq 0.
$$
The corresponding admissible space is the reciprocal image of $[- \infty, \  0]$, which is a closed set of $\mathds{R}$, by a continuous function. Therefore, $(SC)$ defines a closed set. 

iv) The constraint $(LC)$ is derived in
$$    k^L \ l(\omega_i, \  L_i^{(0)}) - A\ell_i^{(0)} \leq 0.
$$
The corresponding admissible space is the reciprocal image of $[- \infty, \  0]$, which is a closed set of $\mathds{R}$, by a continuous function. Therefore, $(LC)$ defines a closed set. 

v) The positivity constraints ($Ax_i^{(0)} \geq 0$, $A\ell_i^{(0)} \geq 0$ and $L_i^{(0)} \geq 0$) also define closed sets, as the reciprocal images of $[0, + \infty]$, which is a closed set of $\mathds{R}$, by a continuous function. 

vi) The constraints $\omega_i \in [0,1]$, $0 \leq \pi_{i,j} \leq 1 - c^{\pi}_j$ , $ 0 \leq\gamma_{i,j} \leq 1 - c^{\gamma}_j $ ($\forall j \in \{1, \dots, n \}$) define a closed admissible set as the reciprocal images of $[0, 1], [0, 1-c^{\pi}_j]$ and $[0, 1-c^{\gamma}_j]$, which are closed sets of $\mathds{R}$, by a continuous function. 
   	
The admissible set $\mathcal{X}_{ad}$ is the intersection of the admissible sets defined by each constraint. Moreover, an intersection of closed sets is a closed set. Thus, combining points i) to vi), we obtain that $\mathcal{X}_{ad}$ is a closed set.			

\subsubsection*{$\mathcal{X}_{ad}$ is a bounded set:}
Let us show that the admissible set is bounded.	\\
Conditions $0 \leq \pi_{i,j} \leq 1 - c^{\pi}_j$ and $ 0 \leq\gamma_{i,j} \leq 1 - c^{\gamma}_j$ ($\forall j \in \{1, \dots, n\}$) show that all the $\pi_{i,j}$ and $\gamma_{i,j}$ are bounded. The same is true for $\omega_i \in [0,1]$. Let us now prove that $A\ell_i$, $Ax_i$ and $L_i$ are bounded. \\

\textbf{i) Bound for $A\ell_i$}\\
The combination of the constraints $L_i^{(0)}\geq 0$ and $(BC)$ implies that the institution invests at least all its own capital. 

\medskip

If $L_i^{(0)}=0$, $K_i^{(0)}$ is an upper-bound for $A\ell_i^{(0)}$. 

Let us now consider the case $L_i^{(0)}>0$. The constraint $(BC)$ can be used to express the debt as a function of other control variables,
$$
L_i^{(0)} = Ax_i^{(0)} + A\ell_i^{(0)} + \Sum_{j=1}^n \pi_{i,j} \mathcal{K}_j^{(0)} + \Sum_{j=1}^n \gamma_{i,j} \mathcal{L}_j^{(0)} - K_i^{(0)}.
$$
Using this last equation, one can express $P_i^{(1)}$ as a function of other control variables:
\begin{align}
P_i^{(1)} &= Ax_i^{(0)} (1+r_i) + A\ell_i^{(0)} (1+r_{rf}) \nonumber
\\& + \Sum_{j=1}^n \pi_{i,j} \max \left( \kappa_j[Ax_j^{(0)} (1+r_j) + A\ell_j^{(0)} (1+r_{rf})] - L_j^{*(1)},0 \right) \nonumber
\\& + \Sum_{j=1}^n \gamma_{i,j}\min \left( \kappa_j [Ax_j^{(0)} (1+r_j) + A\ell_j^{(0)} (1+r_{rf}) ] , L_j^{*(1)} \right) \nonumber
\\& - [ 1+r_D(\omega_i) ] \Big( Ax_i^{(0)} + A\ell_i^{(0)} + \Sum_{j=1}^n \pi_{i,j}K_j^{(0)} + \Sum_{j=1}^n \gamma_{i,j} L_j^{(0)} - K_i^{(0)} \Big) \nonumber
\\& = Al_i^{(0)} [ r_{rf}-r_D(w_i) ] + d(\mathbf{X}_{-A\ell}, \Mb{r}),
\label{Chapnetworks_PnL}
\end{align}
where
$ \mathbf{X}_{-A\ell}=\left( Ax_i^{(0)}, L_i^{(0)}, \omega_i, \pi_{i,1}, \dots, \pi_{i,n}, \gamma_{i,1}, \dots, \gamma_{i,n} \right)^{'}$ is the vector of all the control variables apart from $A\ell_i^{(0)}$, $\Mb{r}=(r_1, \dots,r_n)'$ is the vector of the realized net returns of the external assets and $d$ is some function.
The position $P_i^{(1)}$ is a function of $A\ell_i^{(0)}$, $\mathbf{X}_{-A\ell}$ and $\Mb{r}$, from now on denoted by
$P_i^{(1)}(A\ell_i^{(0)},\mathbf{X}_{-A\ell},\Mb{r})$.
Assumption $(A4)$ states that $r_D(\omega_i)>r_{rf}$, giving that $P_i^{(1)}(.,.,.)$ is strictly decreasing with respect to $A\ell_i^{(0)}$.

Let us consider a value $V_1> K_i^{(0)}$ for $A\ell_i^{(0)}$. From \eqref{Chapnetworks_PnL}, we see that, for all admissible $\mathbf{X}_{-A\ell}$, there exists a set $\varepsilon_1, \dots, \varepsilon_n$ of values such that, if $r_k \geq \varepsilon_k, k=1, \dots, n$, then $P_i^{(1)}(V_1, \mathbf{X}_{-A\ell},\Mb{r}) > 0$. For a second value $V_2$ such that $K_i^{(0)} \leq V_2 < V_1$, we have for all admissible $\mathbf{X}_{-A\ell}$
\Beq
\label{Chapnetworks_Eq_PnLdecrease}
P_i^{(1)} (V_2, \mathbf{X}_{-A\ell},\Mb{r}) > P_i^{(1)} (V_1, \mathbf{X}_{-A\ell},\Mb{r}).
\Eeq
Therefore, if $r_k \geq \varepsilon_k, k=1, \dots, n$, we have
\Beq
\label{Chapnetworks_Positivity_V2}
P_i^{(1)}(V_2, \mathbf{X}_{-A\ell}, \Mb{r}) > 0.
\Eeq

Now, let us compare the expected utility at $A\ell_i^{(0)}=V_1$ and $A\ell_i^{(0)}=V_2$.
We have
\begin{align}
 \mathds{E}\left[ u_i \left( K_i^{(1)} \right) \right](V_1, \mathbf{X}_{-A\ell})
        & = \Int_{-\infty}^{+ \infty} \dots \Int_{-\infty}^{+ \infty}  u_i \left( \max \left[ P_i^{(1)} (V_1, \mathbf{X}_{-A\ell}, \Mb{r}), 0 \right] \right)  f_R(\Mb{r}) \ d\Mb{r} \nonumber \\
        & = \Int_{-\infty}^{\varepsilon_1} \dots \Int_{-\infty}^{\varepsilon_n}  u_i \left( \max \left[ P_i^{(1)}(V_1, \mathbf{X}_{-A\ell}, \Mb{r}), 0 \right] \right) f_R(\Mb{r})\ d\Mb{r} \nonumber \\
        & ~~~~ + \Int_{\varepsilon_1}^{+\infty} \dots \Int_{\varepsilon_n}^{+\infty}  u_i \left[ P_i^{(1)}(V_1, \mathbf{X}_{-A\ell}, \Mb{r}) \right] f_R(\Mb{r})\ d\Mb{r}.
\label{Chapnetworks_Eq_EV1}
\end{align}
By the same decomposition and using \eqref{Chapnetworks_Positivity_V2}, we obtain, for $A\ell_i^{(0)}=V_2 >V_1$,
\begin{align}
 \mathds{E} \left[ u_i \left( K_i^{(1)} \right) \right] (V_2, \mathbf{X}_{-A\ell})
    & = \Int_{-\infty}^{\varepsilon_1} \dots \Int_{-\infty}^{\varepsilon_n}  u_i \left( \max \left[ P_i^{(1)}(V_2, \mathbf{X}_{-A\ell}, \Mb{r}), 0 \right] \right) f_R(\Mb{r})\ d\Mb{r} \nonumber \\
        & ~~~~ + \Int_{\varepsilon_1}^{+\infty} \dots \Int_{\varepsilon_n}^{+\infty}  u_i \left[ P_i^{(1)}(V_2, \mathbf{X}_{-A\ell}, \Mb{r}) \right] f_R(\Mb{r}) \   d\Mb{r}.
\label{Chapnetworks_Eq_EV2}
\end{align}
Using \eqref{Chapnetworks_Eq_PnLdecrease}, we have, for $\Mb{r} \in (- \infty, \varepsilon_1] \times \dots \times (- \infty, \varepsilon_n]$, 
$$ \max \left[ P_i^{(1)} (V_2, \mathbf{X}_{-A\ell}, \Mb{r}), 0 \right] \geq \max \left[ P_i^{(1)} (V_1, \mathbf{X}_{-A\ell}, \Mb{r}), 0 \right],$$
and, since $u_i$ is strictly increasing ($(A2)$),
$$u_i \left( \max \left[ P_i^{(1)} (V_2, \mathbf{X}_{-A\ell}, \Mb{r}), 0 \right] \right) \geq u_i \left( \max \left[ P_i^{(1)} (V_1, \mathbf{X}_{-A\ell}, \Mb{r}), 0 \right] \right).$$
Using \eqref{Chapnetworks_Eq_PnLdecrease}, we have, for all $\Mb{r} \in [\varepsilon_1, + \infty) \times \dots \times [\varepsilon_n, + \infty)$,
$$
P_i^{(1)} (V_2, \mathbf{X}_{-A\ell}, \Mb{r}) > P_i^{(1)} (V_1, \mathbf{X}_{-A\ell}, \Mb{r}),
$$
and, since $u_i$ is strictly increasing,
$$ u_i \left[ P_i^{(1)} (V_2, \mathbf{X}_{-A\ell}, \Mb{r}) \right] > u_i \left[ P_i^{(1)} (V_1, \mathbf{X}_{-A\ell}, \Mb{r}) \right].$$
Moreover, there exists $a \in \mathds{R}$, such that, for all $\Mb{r} \in [a, + \infty)^n, f_R(\Mb{r}) >0$ ($(A3)$).

Therefore, combining \eqref{Chapnetworks_Eq_EV1} and \eqref{Chapnetworks_Eq_EV2} yields
$$  \forall \mbox{ admissible }\ \mathbf{X}_{-A\ell}, ~~~~ \mathds{E} \left[ u_i \left( K_i^{(1)} \right) \right](V_1, \mathbf{X}_{-A\ell})  <  \mathds{E} \left[ u_i \left( K_i^{(1)} \right) \right](V_2, \mathbf{X}_{-A\ell}), 
$$
meaning that, for $A\ell_i^{(0)} \geq K_i^{(0)}$, the objective function is strictly decreasing with respect to $A\ell_i^{(0)}$. Consequently, $\mathcal{P}_i$ is equivalent if we upper-bound the space of $A\ell_i$. Moreover, since $A\ell_i$ is lower-bounded by 0, $A\ell_i$ is bounded.\\

\textbf{ii) Bounds for $Ax_i$ and $L_i$}\\

Let us recall that $(SC)$ is written
$$
 K_i^{(0)} \geq k_i^A Ax_i^{(0)} + k^{\pi} \Sum_{j=1}^n \pi_{i,j} \mathcal{K}_j^{(0)} + k^{\gamma} \Sum_{j=1}^n \gamma_{i,j} \mathcal{L}_j^{(0)}.$$
The equity $K_i^{(0)}$ is fixed as an endowment. Moreover, in the right hand term of $(SC)$, all components are positive. Thus, it imposes that each term is bounded. Therefore, $k_i^A Ax_i$ is upper-bounded and, since $k_i^A >0$ by assumption, $Ax_i$ is upper-bounded. Moreover, $Ax_i \geq 0$ and thus $Ax_i$ is bounded.\\
Using the fact that $k^{\pi} > 0$ and $k^{\gamma} > 0$, we also obtain that both $ \Sum_{j=1}^n \pi_{i,j}K_j^{(0)}$ and $\Sum_{j=1}^n \gamma_{i,j} L_j^{(0)} $ are upper-bounded.
Let us recall that $(BC)$ gives
$$ L_i^{(0)} = Ax_i^{(0)} + A\ell_i^{(0)} + \Sum_{j=1}^n \pi_{i,j} \mathcal{K}_j^{(0)} + \Sum_{j=1}^n \gamma_{i,j} \mathcal{L}_j^{(0)} - K_i^{(0)},$$
implying that $L_i^{(0)}$ is upper-bounded since all terms in the right part of the equation are upper-bounded. Moreover, since $L_i^{(0)} \geq 0$ by assumption, $L_i^{(0)}$ is bounded. 

\subsection*{Existence}

To summarize, the admissible set is not empty.
It is also closed and bounded, and therefore compact. The objective function is continuous and the Weierstrass Theorem ensures the existence of a solution.
\end{proof}

\subsection{For Theorem \ref{PropExistenceUniciteSolOptiProSol}}
\label{AppendixProofExistenceUniciteSolve}

\begin{proof}
~~
\subsection*{Existence}
The existence can be shown exactly in the same way as for Theorem \ref{PropExistenceSolOptiProSol}.

\subsection*{Uniqueness}
The uniqueness is based on a fundamental theorem of optimization: a strictly concave function on a closed convex set admits a unique maximum. We first show that the admissible set is convex and then that the objective function is strictly concave.

\subsubsection*{Convexity of the admissible set}
As before, we denote
$$    \mathbf{X}=\left( Ax_i^{(0)}, A\ell_i^{(0)}, L_i^{(0)}, \omega_i, \pi_{i,1}, \dots, \pi_{i,n}, \gamma_{i,1}, \hdots, \gamma_{i,n} \right)^{'}
    \in \mathcal{X}_{ad},
$$
where $\mathcal{X}_{ad}$ is the admissible space of $\mathcal{P}_i'$. 

Let us show that each constraint of $\mathcal{P}_i'$ defines a convex set. All constraints excluding $(LC)$ involve linear functions of the control variables and thus each of these constraints obviously defines a convex set. 

The constraint $(LC)$ requires more attention. For the sake of notational simplicity, let us denote by $x=\omega_i$, $y=L_i^{(0)}$ and $z=A\ell_i^{(0)}$. The constraint $(LC)$ can therefore be re-written $z > l(x,y)$. The corresponding set is the epigraph of the function $l$. The epigraph is convex if and only if $l$ is convex, i.e. if and only if the Hessian of $l$, $\mathbf{H_l}$, is semi definite positive. By definition, we have
$$
    \mathbf{H_l}=
    \begin{pmatrix}
    \dfrac{\partial^2 l}{\partial x^2} & \dfrac{\partial^2 l}{\partial x \partial y} \\
     \dfrac{\partial^2 l}{\partial x \partial y}  & \dfrac{\partial^2 l}{\partial y^2}
    \end{pmatrix}.
$$
The Sylvester's criterion states that a matrix is semi definite positive if and only if all its leading principal minors are positive, i.e.
$$    \frac{\partial^2 l}{\partial x^2} \geq 0 \ \ \mbox{ and } \ \  
    \frac{\partial^2 l}{\partial x^2} \frac{\partial^2 l}{\partial y^2} \geq \left( \frac{\partial^2 l}{\partial x \partial y} \right)^2.
$$
Thus, under $(A8)$, $(LC)$ defines a convex set and finally all constraints define a convex set. Since the intersection of convex sets is a convex set, $\mathcal{X}_{ad}$ is a convex set.

We want to show that there is uniqueness of the solution of the optimization of the triple $(Ac_i^{(0)}, L_i^{(0)}, \omega_i)$, where $Ac_i^{(0)}$ is one of the variables appearing on the asset side, i.e. among $Ax_i^{(0)}$, $A\ell_i^{(0)}$, $\pi_{i,1}, \dots, \pi_{i,n}, \gamma_{i,1}, \dots, \gamma_{i,n}$.
Let us denote
$$
    \mathbf{X3}=\left (Ac_i^{(0)}, L_i^{(0)}, \omega_i \right)^{'} \in \mathcal{X}3_{ad},
$$
where $\mathcal{X}3_{ad}$ is the admissible set of the three-dimensional optimization program. By using the same arguments as for $\mathcal{X}_{ad}$, $\mathcal{X}3_{ad}$ defines a convex set, whatever the control variable $Ac_i^{(0)}$ that is chosen. Moreover, note that one can show that $\mathcal{X}3_{ad}$ is a closed set, as for Theorem \ref{PropExistenceSolOptiProSol}.


\subsubsection*{Expectation and underlying objective function}
In the following, we generally denote the position by $P_i^{(1)}(\mathbf{X3} , \mathbf{r})$ but sometimes we omit the arguments $\mathbf{X3}$ and $\mathbf{r}$ for simplicity.
The strict concavity of $u_i \left[ v \left( P_i^{(1)} \right) \right]$ is a sufficient condition to obtain the strict concavity of $\mathds{E} \left \{ u_i \left[ v \left( P_i^{(1)} \right) \right] \right \}$ with respect to $\Mb{X3}$.
Indeed, let us assume that  $u_i \left[ v \left( P_i^{(1)} \right) \right]$ is strictly concave.  
Combining the latter assumption with the fact that $f_R$ is strictly positive on $[a, + \infty)^n$ ($(A3)$), we get, for all $(\mathbf{X3}_1, \mathbf{X3}_2) \in \mathcal{X}3_{ad}^2$ and for all $\lambda \in [0,1]$,
\begin{align*}
& \mathds{E} \left \{ u_i \left[ v \left( P_i^{(1)} \right) \right] \right \} (\lambda \mathbf{X3}_1 + (1-\lambda) \mathbf{X3}_2 )
    \\& = \Int_{\mathds{R}^n} u_i \left( v \left[ P_i^{(1)} \left( \lambda \mathbf{X3}_1 + (1-\lambda) \mathbf{X3}_2, \mathbf{r} \right) \right] \right) f_R(\mathbf{r})\  d\Mb{r}  \\
    & > \Int_{\mathds{R}^n} \Bigg[ \lambda\  u_i \left( v \left[ P_i^{(1)}(\mathbf{X3}_1 , \mathbf{r} ) \right] \right) \nonumber 
   + (1-\lambda) \ u_i \left( v \left[ P_i^{(1)}(\mathbf{X3}_2 , \mathbf{r} ) \right] \right) \  \Bigg] f_R(\mathbf{r}) \ d\Mb{r}
\\& = \lambda \ \mathds{E} \left \{ u_i \left[ v \left( P_i^{(1)} \right) \right] \right \}(\mathbf{X3}_1) + (1-\lambda) \mathds{E} \left \{ u_i \left[ v \left( P_i^{(1)} \right) \right] \right \}(\mathbf{X3}_2) \nonumber, 
\end{align*}
showing the strict concavity of the expected utility.

\subsubsection*{Strict concavity of the underlying objective function}
We now focus on $u_i \left[ v \left( P_i^{(1)} \right) \right]$. We consider that only one control variable is free on the asset side.
For the sake of notational simplicity, we denote by $x_1=Ac_i^{(0)}$, $x_2=\omega_i$ and $x_3=L_i^{(0)}$. 
Here we interpret $P_i^{(1)}$ as the function defined by
$$\begin{array}{ccccc}
P_i^{(1)} & : & \mathds{R}^{+} \times [0,1] \times \mathds{R}^{+} & \to & \mathds{R} \\
 & & \begin{pmatrix}
 x_1 \\
 x_2 \\
 x_3
 \end{pmatrix}
  & \mapsto & t(x_1) -[1+r_D(x_2)] x_3,  \\
\end{array}$$
where $t(.)$ is a linear transformation mapping the control variable chosen into the value of the total assets $Ax_i^{(0)} (1+r_i) + Al_i^{(0)} (1+r_{rf}) + \sum_{j=1}^n \pi_{i,j} \mathcal{K}_j^{(0)} + \sum_{j=1}^n \gamma_{i,j} \mathcal{L}_j^{(0)}$.
Let us denote by $g$ the function $u_i \circ v \circ  P_i^{(1)}$. We now study the strict concavity of $g$. We denote by $m=u_i \circ v$, yielding $g=m \circ P_i^{(1)}$.
The function $g$ is strictly concave if and only if its Hessian matrix $\mathbf{H_g}$ is definite negative. We have
$$    \mathbf{H_g}=m''
    \begin{pmatrix}
    1 & -x_3 \ r_D'(x_2) & -[1+r_D(x_2)] \\
    -x_3 \ r_D'(x_2) & -x_3 \left[ \dfrac{m'}{m''} r_D''(x_2) + r_D'^2(x_2) x_3 \right] & r_D'(x_2) \left[ -\dfrac{m'}{m''} + x_3 [1+r_D(x_2)] \right] \\
    -[1+r_D(x_2)] & r_D'(x_2) \left[ -\dfrac{m'}{m''} + x_3 [1+r_D(x_2)] \right] & [1+r_D(x_2)]^2
    \end{pmatrix}.
$$
The Sylvester's criterion states that $\mathbf{H_g}$ is definite negative if and only if all its leading principal minors are strictly negative. Let us now study the three corresponding minors. \\

\textbf{i) First minor}\\
The first minor is
$$    \mbox{Det}_1=|m''|.
$$
According to the Sylvester's criterion, $m''<0$ is imposed. 

\textbf{ii) Second minor}\\
The second minor is
$$
    \mbox{Det}_2=m''^2 \times \left[ -x_3 \left[ \frac{m'}{m''} r_D''(x_2) + r_D'^2(x_2) x_3 \right]  + x_3^2 \ r_D'^2(x_2) \right].
$$
 Thus, the Sylvester's condition imposes, $\forall x_2 \in [0,1]$ and $x_3 \in \mathds{R}^{+}$,
\begin{align*}
                    &  &x_3 \left[ m'   m''  r_D''(x_2) + m''^2  r_D'^2(x_2) x_3 \right] & > x_3^2  r_D'^2(x_2)  m''^2 \\
\Longleftrightarrow &  &m'  m''  r_D''(x_2) + m''^2  r_D'^2(x_2) x_3                     & > x_3  r_D'^2(x_2)  m''^2 \\
\Longleftrightarrow &  &m'  m''  r_D''(x_2)                                               & > 0 \\
\Longleftrightarrow &  &                                         r_D''(x_2)                  & < 0,
\end{align*}
since $m'>0$ by assumption ($u$ and $v$ are strictly increasing so $m=u \circ v$ is strictly increasing as well) and the previous condition (see i) imposes $m''<0$.

\textbf{iii) Third minor}\\
We compute the third minor using Sarrus' rule. We obtain
\begin{align*}
& \mbox{Det}_3= m''^3 \Bigg \{ -x_3 \left[ \frac{m'}{m''} r_D''(x_2) + r_D'^2(x_2) x_3 \right]  [1+r_D(x_2)]^2
\\& +2 (-x_3 \ r_D'(x_2)) \ r_D'(x_2) \left[ -\frac{m'}{m''} + x_3 [1+r_D(x_2)] \right] \  \times (-[1+r_D(x_2)])
\\& - \Bigg[ [1+r_D(x_2)]^2 \ \Bigg(-x_3 \left[ \frac{m'}{m''} r_D''(x_2) + r_D'^2(x_2) x_3 \right] \Bigg)
+ r_D'(x_2) \left[ -\frac{m'}{m''} + x_3 [1+r_D(x_2)] \right]^2
\\& + x_3^2 \ r_D'^2(x_2) \ [ 1+r_D(x_2)]^2 \Bigg]   \Bigg \}
\\&= m''^3 \Bigg[ \left \{ 2 x_3 \ r_D'^2(x_2) [1+r_D(x_2)]+ r_D'^2(x_2) \left( -\frac{m'}{m''}+x_3 \left[ 1+r_D(x_2) \right] \right) \right \} \  \Big(-\frac{m'}{m''}+x_3 [1+r_D(x_2)] \Big)
\\& + x_3^2 \ r_D'^2(x_2)) \ [1+r_D(x_2)]^2 \Bigg].
\end{align*}
Considering $m''<0$ (see ii) and $m'>0$ (by assumption), we have $\dfrac{m'}{m''}<0$. Thus, assuming $\forall x_2 \in [0,1] , r_D'(x_2) \ne 0$, all terms in the brackets are strictly positive. Moreover $m''^3<0$ and thus the condition $\mbox{Det}_3 <0$ is satisfied.

\subsubsection*{Summary}
The following assumptions
\begin{itemize}
    \item $m''(x)<0$  ($(A5)$);
    \item $r_D''< 0$ ($(A6)$);
    \item $r_D' \ne 0$ ($(A7)$);
\end{itemize}
are sufficient to ensure that the Hessian matrix of $g$ is definite negative and  therefore that $g$ is strictly concave with respect to the control variable $Ac_i^{(0)}$, the debt $L_i^{(0)}$ and the maturity transformation $\omega_i$.

Finally, under $(A5)$, $(A6)$, $(A7)$ and $(A8)$, the objective function 
$\mathds{E} \left \{ u_i \left[ v \left( P_i^{(1)} \right) \right] \right \}$ is strictly concave on a closed convex set, showing the uniqueness.
\end{proof}

\medskip 

\begin{Rq}
Let us now come back to the choice of working directly on the integrand. Even if $u_i \circ v$ is not strictly concave everywhere, one may certainly expect the strict concavity to come from the integration with respect to the realized returns $\mathbf{r}$ (for some appropriate densities $f_R$). However, as we have shown, studying the concavity of a multivariate function involves studying its Hessian and this is already quite complicated in the case of the integrand. The Hessian matrix of the expected utility implies much more complicated expressions, especially products of integral, apart from the first leading minor. The condition on this first leading minor is written as follows:
$$
    \Int_{\mathds{R}^n} (u_i \circ v \circ P_i^{(1)})''(\Mb{X}, \mathbf{r}) f_R(\mathbf{r})  \ d\mathbf{r} > 0.
$$
Thus, even in the case of the first leading minor, it seems difficult to obtain results except in particular cases of very simple density functions $f_R$. Moreover, the study of the uniqueness of all control variables (and thus the study of the strict concavity with respect to all control variables) would require the study of a high dimensional Hessian, which is very difficult.
\end{Rq}


\subsection{For Lemma \ref{CorV}}

\begin{proof}
\textbf{i)} We consider the function defined by $\forall P \in \mathds{R}, v(P)=P$. \\
We have
$v'(P)=1$ and $v''(P)=0$. Thus, $(u_i \circ v)'(P)=u_i'[v(P)] v'(P)=u_i'(P)$, giving $(u_i \circ v)''(P)=u_i''(P)$. Therefore, $(A5)$ imposes $\forall P \in \mathds{R}, u_i''(P) <0$.

\medskip

\textbf{ii)}  Here, we consider the function defined by $\forall P \in \mathds{R}, v(P)= \log \left( \exp(P)+1 \right)$. \\
We have, for all $P \in \mathds{R}$,
$$ v'(P)=\frac{e^P}{e^P+1} \mbox{ and } v''(P)=\frac{e^P (e^P+1) -e^P\ e^P }{(e^P+1)^2} = \frac{e^P}{(e^P+1)^2}.$$
Let us study the function $h=u_i \circ v = \log \circ v$. We have
$$ h'(P)=\frac{v'(P)}{v(P)},$$
and thus
\begin{align*}
h''(P) &=\frac{v''(P)\ v(P) - v'^2(P)}{v^2(P)}
\\& = \frac{e^P}{(e^P+1)^2} \frac{1}{\log \left( e^P +1 \right)} - \frac{e^P e^P}{(e^P+1)^2} \frac{1}{ [\log \left( e^P +1 \right)]^2 }
\\& = \frac{e^P}{(e^P+1)^2} \frac{1}{\log \left( e^P +1 \right)} \left( 1-\frac{e^P}{\log \left( e^P +1 \right)} \right).
\end{align*}
The first two factors are positive whereas the third one is negative (since  $\forall x \in \mathds{R}_{+}^*, \log(1+x) < x$). Consequently, $\forall P \in \mathds{R}, \ h''(P) < 0$. Hence the result.
\end{proof}

\subsection{For Lemma \ref{CorL}}

\begin{proof}
We consider the function defined by $\forall \omega \in [0,1] \mbox{ and } \forall L \in \mathds{R}^+, l(\omega, L)=\exp(\omega) \exp(L)$. We have
$$  \forall \omega \in [0,1] \mbox{ and } \forall L \in \mathds{R}^+, \dfrac{\partial^2 l}{\partial \omega^2}= \exp(\omega) \exp(L)>0 \ \ \ \mbox{ and }$$
$$\frac{\partial^2 l}{\partial \omega^2} \frac{\partial^2 l}{\partial L^2}= [\exp(\omega) \exp(L)]^2= \left( \frac{\partial^2 l}{\partial \omega \partial L}  \right)^2.$$
That shows that $(A8)$ is satisfied.
\end{proof}

\subsection{For Proposition \ref{Prop_PositivenessPi}}

\begin{proof}
The proof is based on the Karuch, Kuhn, Tucker (KKT) Theorem, which provides necessary conditions on a local optimum of an optimization problem under equality and inequality constraints. We show that assuming $\pi^*=0$ leads to a contradiction. \\

The KKT Theorem states that there exist coefficients $\mu_i \geq 0$ such that a local maximum $(Ax^*, \pi^*)$ is a local maximum of the objective function $\mathcal{L}^a$, defined as
    $$ \mathcal{L}^a = f - \mu_1 (k^A Ax + k^{\pi} \pi -1) + \mu_2 Ax + \mu_3 \pi - \mu_4 (\pi-1), $$
where $f$ is the initial objective function, i.e. $\mathds{E}  \left[ u( Ax R_g + \pi R_g^{\pi} ) \right]$. Moreover, the $\mu_i$ coefficients satisfy
$$\forall i, \mu_i C_i=0,$$
where $C_i$ is the $i$-th constraint. \\

At a local optimum, the KKT conditions are
\begin{equation*}
	\left\{
	\begin{array}{l}
	\Frac{\partial f}{\partial Ax}  - \mu_1 k^{A} + \mu_2=0 \\
	\Frac{\partial f}{\partial \pi} + \mu_3 - \mu_1 k^{\pi} - \mu_4 =0 \\
	\mu_1 (k^A Ax^* + k^{\pi} \pi^* -1)=0 \\
	\mu_2 Ax^* =0 \\
	\mu_3 \pi^* =0 \\
	\mu_4 (\pi^* -1)=0 \\
    \end{array}
    \right..
\end{equation*}
We now assume that $\pi^*=0$. The last equation directly provides $\mu_4=0$. Since $f$ is strictly increasing, $Ax^*$ is necessarily strictly positive (such $Ax^*$ is compatible with the constraints). Therefore, we have $\mu_2=0$. Thus, the first equation provides
    $$ \mu_1 = \Frac{\partial f}{\partial Ax} \Frac{1}{k^{A}}.$$
Injecting this result into the second equation gives
 \Beq
 \label{Chapnetworks_Eq_KKT_mu3}
 \mu_3 = \Frac{\partial f}{\partial Ax} \Frac{k^{\pi}}{k^{A}}  - \Frac{\partial f}{\partial \pi} < 0 \ (\mbox{by assumption}).
 \Eeq
Equation \eqref{Chapnetworks_Eq_KKT_mu3} is in contradiction with the KKT theorem, stating that $\forall i, \mu_i \geq 0$. Therefore, $\pi^*\neq0$.
\end{proof}

\subsection{For Proposition \ref{PropPositiveness_1}}

\begin{proof}
First, let us recall that
\begin{align*}
P_i^{(1)}
& = Ax_i^{(1)} + A\ell_i^{(1)}+ \Sum_{j=1}^n \pi_{i,j} K_j^{(1)} + \Sum_{j=1}^n \gamma_{i,j} L_j^{(1)}
\\& - [1+r_D(\omega_i)] \Big( Ax_i^{(0)} + A\ell_i^{(0)} + \Sum_{j=1}^n \pi_{i,j}\mathcal{K}_j^{(0)} + \Sum_{j=1}^n \gamma_{i,j} \mathcal{L}_j^{(0)} - K_i^{(0)} \Big)
\\& = Ax_i^{(1)} + A\ell_i^{(1)} - [1+r_D(\omega_i)] (Ax_i^{(0)} + A\ell_i^{(0)} - K_i^{(0)}) + \Sum_{j=1}^n \pi_{i,j} \Big( K_j^{(1)} - [1+r_D(\omega_i)]\  \mathcal{K}_j^{(0)} \Big)
\\& + \Sum_{j=1}^n \gamma_{i,j} \Big( L_j^{(1)} - [1+r_D(\omega_i)] \mathcal{L}_j^{(0)} \Big).
\end{align*}

We have $u_i \circ v=Id$. Then the derivative of the objective function with respect to $\pi_{i,j}$ is written
\Beq
\label{Eq_Derivative_Share}
\frac{\partial \mathds{E}[u_i \circ v(P_i^{(1)})]}{\partial \pi_{ij}} = \frac{\partial \mathds{E}(P_i^{(1)})}{\partial \pi_{ij}} = \mathds{E} \left[ K_j^{(1)} - [1+r_D(\omega_i)]\  \mathcal{K}_j^{(0)} \right] = \mathds{E} \left[ K_j^{(1)} \right]- [1+r_D(\omega_i)]\  \mathcal{K}_j^{(0)},
\Eeq
where
$$ K_j^{(1)} = \max \left( \kappa_j \left( Ax_j^{(1)} + A\ell_j^{(1)} \right) - L_j^{*(1)}, 0 \right).$$
Let us now explicit the latter expression:
\begin{align*}
\kappa_j \left( Ax_j^{(1)} + A\ell_j^{(1)} \right) - L_j^{*(1)} &= \kappa_j \left( Ax_j^{(0)}(1+r_j) + A\ell_j^{(0)} (1+r_{rf}) \right) - L_j^{*} [1+r_D(\omega_j)]
\\& = a_j r_j + b_j,
\end{align*}
by denoting $a_j=\kappa_j Ax_j^{(0)}$ and $b_j=\kappa_j \left( Ax_j^{(0)}+ A\ell_j^{(0)} (1+r_{rf}) \right) - L_j^* [1+r_D(\omega_j)]  $ \\
Then,
\Beq
\label{Eq_Equity_1}
\mathds{E} \left[ K_j^{(1)} \right] =  \mathds{E} \left[  \max \left(a_j r_j + b_j, 0 \right) \right] = \int_{\frac{-b_j}{a_j}}^{+ \infty} (a_j r_j + b_j) \ f_{R,j}(r_j) \  dr_j.
\Eeq
Combining \eqref{Eq_Derivative_Share} and \eqref{Eq_Equity_1}, we obtain
$$ \frac{\partial \mathds{E}[u_i \circ v(P_i^{(1)})]}{\partial \pi_{ij}} >0 \Longleftrightarrow \int_{\frac{-b_j}{a_j}}^{+ \infty} (a_j r_j + b_j) \ f_{R,j}(r_j) > [1+r_D(\omega_i)]\  \mathcal{K}_j^{(0)}.$$

The derivative with respect to $\gamma_{ij}$ is written
\begin{align*}
\frac{\partial \mathds{E}(P_i^{(1)})}{\partial \gamma_{ij}} 
&= \mathds{E}[L_j^{(1)}] - [1+r_D(\omega_i)] \mathcal{L}_j^{(0)} 
\\& = \mathds{E} \left[ \min \left(a_j r_j + b_j, L_j^{*(1)}  \right) \right] - [1+r_D(\omega_i)] \mathcal{L}_j^{(0)}
\\& = \int_{- \infty}^{\frac{L_j^{*(1)}-b_j}{a_j}} (a_j r_j + b_j) f_{R,j}(r_j) \  dr_j + L_j^{*(1)} \mathds{P} \left( r_j > \frac{L_j^{*(1)}-b_j}{a_j} \right) - [1+r_D(\omega_i)] \mathcal{L}_j^{(0)}
 \\& = \int_{- \infty}^{\frac{L_j^{*(1)}-b_j}{a_j}} (a_j r_j + b_j) f_{R,j}(r_j) \  dr_j  + L_j^{*} c_j\  [1+r_D(\omega_j)] - [1+r_D(\omega_i)] \mathcal{L}_j^{(0)},
\end{align*}
by denoting $c_j=\mathds{P} \left( r_j > \frac{L_j^{*(1)}-b_j}{a_j} \right)$.
Finally,
$$ \frac{\partial \mathds{E}[u_i \circ v(P_i^{(1)})]}{\partial \pi_{ij}} >0 \Longleftrightarrow \int_{- \infty}^{\frac{L_j^{*}[1+r_D(\omega_j)]-b_j}{a_j}} (a_j r_j + b_j) f_{R,j}(r_j) \  dr_j  + L_j^{*} c_j\  [1+r_D(\omega_j)] > [1+r_D(\omega_i)] \mathcal{L}_j^{(0)}.$$
\end{proof}

\subsection{For Proposition \ref{PropPositiveness_2}}

\begin{proof}
Let us at first consider the derivative with respect to $\pi_{ij}$.
We have
$$ \frac{\partial \left[ u_i \circ v \left( P_i^{(1)} \right) \right]}{\partial \pi_{ij}} = \frac{\partial (u_i \circ v)}{\partial P_i^{(1)}} \frac{\partial P_i^{(1)}}{\partial \pi_{ij}}.$$
The first term $\Frac{\partial (u_i \circ v)}{\partial P_i^{(1)}}$ can be interpreted as some kind of marginal utility (with the utility being composed with function $v$). It depends on the returns of all banks connected to Bank $i$ and not only on the return of Bank $j$. Let us denote
$$ h_{i1}(r_1, \dots, r_j, \dots, r_n)=\frac{\partial (u_i \circ v)}{\partial P_i^{(1)}}.$$
Moreover, we have
$$ \frac{\partial P_i^{(1)}}{\partial \pi_{ij}} = K_j^{(1)} - [1+r_D(\omega_i)]\  \mathcal{K}_j^{(0)} = \max \left(a_j r_j + b_j, 0 \right) - [1+r_D(\omega_i)]\  \mathcal{K}_j^{(0)}.$$
Let us introduce
$$ h_{i2}(r_j)=\frac{\partial P_i^{(1)}}{\partial \pi_{ij}}.$$
Thus,
\begin{align*}
&\frac{\partial \mathds{E} \left [ u_i \circ v \left( P_i^{(1)} \right)  \right ] }{\partial \pi_{ij}}
\\& = \mathds{E} \left[ \frac{\partial u_i \circ v \left( P_i^{(1)} \right)}{\partial \pi_{ij}} \right]
\\& = \int_{r_1=-\infty}^{+ \infty} \dots \int_{r_j} \dots \int_{r_n}  h_{i1}(r_1, \dots, r_j, \dots, r_n) \ h_{i2}(r_j) f_R(r_1, \dots, r_n) \ dr_n \dots dr_j \dots dr_1
\\& = \int_{r_j=-\infty}^{+ \infty} \left[ \int_{r_1}  \dots \int_{r_{j-1}} \int_{r_{j+1}} \dots \int_{r_n}  h_{i1}(r_1, \dots, r_j, \dots, r_n) h_{i2}(r_j)f_R(r_1, \dots, r_n) \ dr_n \dots dr_{j+1} \ dr_{j-1} \dots dr_1 \right] dr_j
\\& = \int_{r_j=-\infty}^{+ \infty} h_{i2}(r_j)  \left[\int_{r_1}  \dots \int_{r_{j-1}} \int_{r_{j+1}} \dots \int_{r_n}  h_{i1}(r_1, \dots, r_j, \dots, r_n) f_R(r_1, \dots, r_n) \ dr_n \dots dr_{j+1} \ dr_{j-1} \dots dr_1 \right] dr_j
\\& = \int_{- \infty}^{+ \infty} h_{i2}(r_j) w(r_j) \ dr_j
\\& = \int_{-\frac{b_j}{a_j}}^{+ \infty} (a_j r_j + b_j) w(r_j) \ dr_j - \int_{- \infty}^{+ \infty} [1+r_D(\omega_i)] \mathcal{K}_j^{(0)} w(r_j)\  dr_j
\\& =  \int_{-\frac{b_j}{a_j}}^{+ \infty} (a_j r_j + b_j) w(r_j) \ dr_j - [1+r_D(\omega_i)] \mathcal{K}_j^{(0)} \int_{- \infty}^{+ \infty} w(r_j) \ dr_j,
\end{align*}
where
\begin{align*}
& w(r_j)
\\ &= \int_{r_1=-\infty}^{+ \infty}  \dots \int_{r_{j-1}} \int_{r_{j+1}} \dots \int_{r_n}  h_{i1}(r_1, \dots, r_j, \dots, r_n) \ f_R(r_1, \dots, r_n) \ dr_n \dots dr_{j+1} \ dr_{j-1} \dots dr_1.
\end{align*}

Therefore,
$$\frac{\partial \mathds{E} \left [ u_i \circ v \left( P_i^{(1)} \right)  \right ] }{\partial \pi_{ij}} >0 \Longleftrightarrow \int_{-\frac{b_j}{a_j}}^{+ \infty} (a_j r_j + b_j) w(r_j) \ dr_j > [1+r_D(\omega_i)] \mathcal{K}_j^{(0)} \int_{- \infty}^{+ \infty} w(r_j) \ dr_j.$$

\medskip

Let us now consider the case of $\gamma_{ij}$.
As in the previous case, the corresponding derivative is written
$$ \frac{\partial \left[ u_i \circ v \left( P_i^{(1)} \right) \right]  }{\partial \gamma_{ij}} = \frac{\partial (u \ o\  v)}{\partial P_i^{(1)}} \frac{\partial P_i^{(1)}}{\partial \gamma_{ij}}.$$
The first term is equal to $h_{i1}(r_j)$ and the second is denoted
$h_{i3} (r_j).$
The same computation as in the case of $\pi_{ij}$ yields
$$ \frac{\partial \mathds{E} \left [ u_i \circ v \left( P_i^{(1)} \right) \right] }{\partial \gamma_{ij}} = \int_{- \infty}^{+ \infty} h_{i3}(r_j) w(r_j) \ dr_j.$$
We have
$$ h_{i3}(r_j)=L_j^{(1)}-[1+r_D(\omega_i)] \mathcal{L}_j^{(0)}=\min(a_j r_j+d_j, L_j^{*}[1+r_D(\omega_j)]) - [1+r_D(\omega_i)] \mathcal{L}_j^{(0)},$$
where $d_j=\kappa_j \left( Ax_j^{(0)}+ A\ell_j^{(0)} (1+r_{rf}) \right)$.
Finally
\begin{align*}
\frac{\partial \mathds{E} \left [ u_i \circ v \left( P_i^{(1)} \right) \right] }{\partial \gamma_{ij}} &= \int_{-\infty}^{\frac{L_j^*[1+r_D(\omega_j)]-d_j}{a_j}} (a_j r_j + d_j) w(r_j) \  dr_j+ L_j^{*}[1+r_D(\omega_j)]) \int_{\frac{L_j^*[1+r_D(\omega_j)]-d_j}{a_j}}^{+ \infty} w(r_j) \  dr_j
\\& - [1+r_D(\omega_i)] \mathcal{L}_j^{(0)} \int_{- \infty}^{+ \infty}  w(r_j) \  dr_j,
\end{align*}
giving that
\begin{align*}
& \frac{\partial \mathds{E} \left [ u_i \circ v \left( P_i^{(1)} \right) \right] }{\partial \gamma_{ij}} >0
\\& \Longleftrightarrow \int_{-\infty}^{\frac{L_j^*[1+r_D(\omega_j)]-d_j}{a_j}} (a_j r_j + d_j) w(r_j) \  dr_j+ L_j^{*}[1+r_D(\omega_j)]) \int_{\frac{L_j^*[1+r_D(\omega_j)]-d_j}{a_j}}^{+ \infty} w(r_j) \  dr_j 
\\& \ \ \ \ \ > [1+r_D(\omega_i)] \mathcal{L}_j^{(0)} \int_{- \infty}^{+ \infty}  w(r_j) \  dr_j.
\end{align*}
\end{proof}

\subsection{For Proposition \ref{Prop_PricingMarketValue}}

\begin{proof}
Recall that we consider the following dynamics for $Ax_i, i=1,2$:
$$ \log \left( \Frac{Ax_i^{(1)}}{Ax_i^{(0)}}\right) \sim \mathcal{N}(\mu_i, \sigma_i)  \ \ \mbox{ i.e. } Ax_i^{(1)} = Ax_i^{(0)} e^{\mu_i + \sigma_i U}, \ \mbox{where } U \sim \mathcal{N}(0, 1). $$
We have $ K_i^{(1)} = \max( \kappa_i Ax_i^{(1)} - L_i^{*(1)}, 0)$ and $ L_i^{(1)} = \min( \kappa_i Ax_i^{(1)}, L_i^{*(1)})$. \\
We define $\tilde{u}$ such that $ \kappa_i Ax_i^{(0)} e^{\mu_i + \sigma_i \tilde{u}} = L_i^{*(1)}$, i.e. $ \tilde{u} = \Frac{1}{\sigma_i}\left( \log \left( \Frac{L_i^{*(1)}}{\kappa_i Ax_i^{(0)}} \right) - \mu_i \right).$
Let us denote by $\phi$ the density of the standard Gaussian variable. We have
\begin{align*}
    \mathds{E}_0 \left(K_i^{(1)} \right) & = \mathds{E}_0 \left[ \max( \kappa_i Ax_i^{(1)} - L_i^{*(1)}, 0)  \right]  
    \\& = \Int_{\tilde{u}}^{+\infty} \left( \kappa_i Ax_i^{(0)} e^{\mu_i + \sigma_i u}-L_i^{*(1)} \right) \phi(u) \ du
    \\& = \kappa_i Ax_i^{(0)} e^{\mu_i} \Int_{\tilde{u}}^{+\infty} e^{\sigma_i u} \phi(u)\ du - L_i^{*(1)} \Int_{\tilde{u}}^{+\infty} \phi(u)\ du \\
            & = \kappa_i Ax_i^{(0)} e^{\mu_i} \Int_{\tilde{u}}^{+\infty} e^{\sigma_i u} \Frac{1}{\sqrt{2\pi}} e^{-\frac{1}{2}u^2} \ du  - L_i^{*(1)} \left[ 1 - \Phi(\tilde{u}) \right] \\
            & = \kappa_i Ax_i^{(0)} e^{\mu_i} \Int_{\tilde{u}}^{+\infty} e^{\frac{1}{2}\sigma_i^2 } \Frac{1}{\sqrt{2\pi}} e^{-\frac{1}{2}(u-\sigma_i)^2} du  - L_i^{*(1)} \left[ 1 - \Phi(\tilde{u}) \right] \\
            & = \kappa_i Ax_i^{(0)} e^{\mu_i + \frac{1}{2}\sigma_i^2 } \Int_{\tilde{u}-\sigma_i }^{+\infty} \phi(v)\ dv - L_i^{*(1)} \left[ 1 - \Phi(\tilde{u}) \right] \mbox{ (by the change of variable } v=u-\sigma_i)
            \\& = \kappa_i Ax_i^{(0)} e^{\mu_i  + \frac{1}{2}\sigma_i^2 } \left[ 1 - \Phi(\tilde{u}-\sigma_i) \right]-   L_i^{*(1)} \left[ 1 - \Phi(\tilde{u}) \right].
\end{align*}
In the same way,
\begin{align*}
    \mathds{E}_0\left(L_i^{(1)} \right) & = \mathds{E}_0\left[  \min( \kappa_i Ax_i^{(1)}, L_i^{*(1)})  \right]
    \\
        & = \Int_{-\infty}^{\tilde{u}} \kappa_i Ax_i^{(0)}e^{\mu_i + \sigma_i u} \phi(u)\ du  + \Int_{\tilde{u}}^{+\infty} L_i^{*(1)} \phi(u)\ du \\
        & = \kappa_i Ax_i^{(0)} e^{\mu_i  + \frac{1}{2}\sigma_i^2 } \Int_{-\infty}^{\tilde{u}-\sigma_i}  \phi(v)\ dv + L_i^{*(1)}   \left[ 1 - \Phi(\tilde{u}) \right] 
        \\& \ \ \ \mbox{ (using the same trick than in the case of } K_i^{(1)})
        \\
        & = \kappa_i Ax_i^{(0)} e^{\mu_i  + \frac{1}{2}\sigma_i^2 } \Phi\left( \tilde{u}-\sigma_i  \right) + L_i^{*(1)}   \left[ 1 - \Phi(\tilde{u}) \right].
\end{align*}
\end{proof}

\section{Algorithm of network formation}
\label{AppendixAlgoNetworkFormation}

In the case of 2 institutions ($n=2$), the algorithm of network formation is the following:
\begin{enumerate}
\item \textbf{Optimization for Bank 1 without interconnections}. Indeed in this first step, $K_2^{(0)}$ and $L_2^{(0)}$ are not known. \\
We then have to optimize
$E \left \{ u_1 \left[ v \left( P_1^{(1)}(Ax_1^{(0)}, L_1^{(0)}) \right) \right] \right \},$
where
\begin{align*}
P_1^{(1)} &= Ax_1^{(0)} (1+r_1) - [1+r_{D,1}] L_1^{(0)}.
\end{align*}
This step provides $Ax_1^{(0)}$ and $L_1^{(0)}$.
\item \textbf{Optimization for Bank 2 with interconnections}. We have
\begin{align*}
& P_2^{(1)} =  Ax_2^{(0)} (1+r_2)
		+ \pi_{2,1} \max \left( \kappa_1 Ax_1^{(0)} (1+r_1) - L_1^{(0)} \  [1+r_{D,1}], 0 \right)
		\\& + \gamma_{2,1} \min \left( \kappa_1 Ax_1^{(0)} (1+r_1), L_1^{(0)} [1+r_{D,1}] \right)
- [1+r_{D,2}] L_2^{(0)},
\end{align*}
where
$ \kappa_1 = \dfrac{L_1^{(0)} + K_1^{(0)} }{ Ax_1^{(0)}}
$
is the scaling factor compensating the absence of interconnections (it keeps the balance sheet of Bank $1$ balanced).
Since $K_1^{(0)}$ has been obtained at step 1 under the assumption that Bank 1 is not interconnected, here $\kappa_1 =1$. But this will be corrected in further iterations.\\
This step gives $Ax_2^{(0)}, L_2^{(0)}, \pi_{2,1}$ and $\gamma_{2,1}$.
\item \textbf{Optimization for Bank 1 with interconnections}. We have
\begin{align*}
& P_1^{(1)} = Ax_1^{(0)} (1+r_1) 
		+ \pi_{1,2} \max \left( \kappa_2 Ax_2^{(0)}(1+r_2) - L_2^{(0)} [1+r_{D,2}], 0 \right)
		\\& + \gamma_{1,2} \ \min \left( \kappa_2 Ax_2^{(0)} (1+r_2), L_2^{(0)} [1+r_{D,2})] \right)
- [1+r_{D,1}] L_1^{(0)}
\end{align*}
where $
\kappa_2 = \dfrac{L_2^{(0)} + K_2^{(0)} }{ Ax_2^{(0)}}.
$
This step provides $Ax_1^{(0)}, L_1^{(0)}, \pi_{1,2}$ and $\gamma_{1,2}$.
\item \textbf{New optimization for Bank 2 with interconnections}. \\
Note that at this step, $\kappa_1 = \dfrac{L_1^{(0)} + K_1^{(0)} }{ Ax_1^{(0)}} > 1$, since at the previous step, the optimization has be done for Bank 1 being interconnected.
\item \textbf{New optimization for Bank 1 with interconnections}.
\end{enumerate}
Further iterations can be carried out if the variation in the estimates from one step to the next is higher than a predefined threshold.

\section{Calibration of external assets returns}
\label{AppendixDynamicCalibration}

Given the values of the mean net returns and the probability of default, let us derive the corresponding values of $\mu_i$ and $\sigma_i$, for $i = 1, 2$.
We denote by $GR_i$ and $NR_i$ the gross and the net return of Bank $i$, respectively. They satisfy the relationship $NR_i=GR_i-1$.
Thus, since the gross returns are log-normal, $$ \mathds{E}(NR_i) =\mathds{E}(GR_i)-1 =\exp \left( \mu_i+\frac{\sigma_i^2}{2} \right)-1 .$$
If we denote by $m_i$ the empirical mean of the net return, we then have
$$ m_i =\exp \left( \mu_i+\frac{\sigma_i^2}{2} \right)-1,$$
that gives
\Beq
\label{SystEq1}
\mu_i= \log(1+m_i)- \frac{\sigma_i^2}{2}.
\Eeq

We need a second equation to find $\mu_i$ and $\sigma_i$. We could use the expression
$$ \mbox{Var}(NR_i) =\mbox{Var}(GR_i) = (\exp(\sigma_i^2)-1) \ \exp(2\mu_i+\sigma_i^2) $$
i.e., by denoting $v_i$ the empirical variance of $RN_i$,
$$  v_i = (\exp(\sigma_i^2)-1) \ \exp(2\mu_i+\sigma_i^2).$$
However, it is difficult to find reliable values for $v_i$. If we consider banks' data, only one return is available per year and thus the estimation of the variance is inaccurate. Another possibility is to compute the variance of the net returns of an index like the CAC 40. However, such an index is not representative of the external assets of a financial institution since it only contains shares. Moreover, it does not take the hedging strategy of the institution into account.

Therefore, we choose to derive the needed equation from the probability of default. This quantity is indeed easier to obtain. Actually, the usual rating for large banks corresponds to a probability of default of about $0.1\%$. Considering an autarkic stylized bank with debt $L_i$ and a total asset $A_i$, whose gross returns are log-normal of parameter $(\mu_i,\sigma_i)$, the probability of default is
\begin{equation}
\label{Systeq2}
    PD = \Phi\left( \Frac{ \log\left( \Frac{L_i}{A_i} \right) - \mu_i}{\sigma_i} \right).
\end{equation}
Using in \eqref{Systeq2} the expression of $\mu_i$ in \eqref{SystEq1}, we obtain
$$
PD = \Phi\left( \Frac{ \log\left( \Frac{L_i}{A_i} \right) - \log(1+m_i) + \Frac{\sigma_i^2}{2} }{\sigma_i} \right) = \Phi\left( \Frac{ \log\left( \Frac{L_i}{A_i (1+m_i)} \right) + \Frac{\sigma_i^2}{2} }{\sigma_i} \right).
$$
If we denote by $p$ the empirical probability of default, the equation to solve is
\begin{align*}
& p = \Phi\left( \Frac{ \log\left( \Frac{L_i}{A_i (1+m_i)} \right) + \Frac{\sigma_i^2}{2} }{\sigma_i} \right) \Longleftrightarrow & \Frac{\sigma_i^2}{2} - \sigma_i \  \Phi^{-1} (p) + \log\left( \Frac{L_i}{A_i (1+m_i)} \right) =0.
\end{align*}
This is a quadratic equation with discriminant $\Delta= [\Phi^{-1}(p)]^2 - 2 \log\left( \Frac{L_i}{A_i (1+m_i)} \right)$.
With chosen values of $A_i$, $L_i$ and $m_i$, we know that $\Delta >0$ and thus $ \sigma_i=\Phi^{-1}(p)+\sqrt{\Delta}$, since the other solution is strictly negative and thus unsuitable for a volatility. Finally the implied volatility is written
\begin{equation}
    \sigma_i =\Phi^{-1}(p)+\sqrt{[\Phi^{-1}(p)]^2 - 2 \log\left( \Frac{L_i}{A_i (1+m_i)} \right)}.
\end{equation}
We then obtain $\mu_i$ using \eqref{SystEq1}. \\

\section{Algorithm of equilibrium computation}
\label{AppendixAlgoWelfareComputation}

The computation of the equilibrium involving $n$ financial institutions requires to solve up to $2^n$ linear systems with a brutal force approach (see \cite{gourieroux2012bilateral} for details), implying a total complexity in $O(n^3 \times 2^n)$. The cubic term stems from the resolution of a linear system that requires to invert a $n \times n$ matrix. Only a little gain can be obtained on this term. The exponential term comes from testing each possible situation: each institution is either solvent or in default.

\medskip

Instead, in order to deal with the exponential term, we adopt an heuristic algorithm. The key idea is to test the $2^n$ potential regimes in a "proper" order and to use the existence and uniqueness property to stop the algorithm as soon as one feasible solution is computed. 
Since interconnections are small, the way of sorting the regimes relies on the situation without interconnections.

To do so, let us define Regime $r$ by $\Mb{d}^r=(d_1^r,\dots,d^r_n)'$, where $d^r_i=-1$ if Institution $i$ is in default and $1$ otherwise (for $i=1,\dots,n$). We define a weight vector $\Mb{w}=(w_1,\dots,w_n)$, where $w_i = (Ax_i + A\ell_i - L_i^*) / L_i^*$ (for $i=1,\dots,n$). Note that $\Mb{w}$ depends on some known inputs only and can therefore be easily computed. When $w_i$ is positive, the external assets of financial institution $i$ are higher than its nominal debt. Therefore, whatever the situations of other financial institutions, Institution $i$ is always solvent at the equilibrium. On the contrary, when $w_i$ is negative, the financial institution needs a sufficient amount of inter-financial assets to be solvent. In that case, since interconnections are assumed to be small, the (absolute) value of $w_i$ indicates the likelihood (in a \emph{non-statistical} sense) of default of Institution $i$. One can associate to Regime $r$ a score, given by $\Mb{w}.\Mb{d}^r$, which
measures the likelihood of Regime $r$. 
For instance, if $\Mb{w}$ contains many negative values, we might think that the equilibrium lies in a regime with a lot of institutions in default. Thus, a regime with many values of $d^r_i$ equal to $-1$ will be likely and will be associated to a high score.

Actually, the regime with the highest score can easily be derived from $\Mb{w}$.  This regime, labeled $\underline{r}$, is defined by $d_i^{\underline{r}} = \Mb{I}_{ \{ w_i > 0 \} } - \Mb{I}_{ \{ w_i \leq 0 \} }$, for $i=1,\dots,n$. 
If $w_i \leq 0$, it is likely that Institution $i$ is in default and thus we set $d_i^{\underline{r}} =-1$. The contrary is true when $w_i > 0$. 
We test this most likely regime. If it corresponds to the solution, we have finished. If not, 
one can switch the components of $\Mb{d}^{\underline{r}}$  one by one to get new regimes with high scores. It is important to keep in mind that assuming the default of an institution with positive weight is dead-end. While no solution was found, this mechanism of building new regimes 
can be carried on until having sorted all the potential regimes apart from the ones for which there exists $i$ such that $w_i>0$ and $d_i=-1$. 

The complexity (in the worst case) of this algorithm is in $O(n^3 \times 2^{n-p})$, where $p=\# \{ i : w_i > 0\}$, with $\#$ standing for the cardinal. Thus, we still have an exponential term. However, the expectation of the number of regimes to be tested before finding the solution is much lower than in the case of the brutal force approach. The algorithm performs well in practice. For example, with $10$ financial institutions having log-normal returns and random interconnections, the equilibrium lies in the $10$ first tested regimes in almost all cases. \\

NB: If one remains concerned by exploring all the regimes (implying keeping the exponential term in the complexity), one solution is to stop the search after an arbitrary number of regimes (for instance $n$). When the exploration approach is stopped, a pure numerical approach can be carried out, in order to solve the system 
$(\Mb{K}, \Mb{L})^{'}=q[(\Mb{K}, \Mb{L})^{'}]$, where the function $q$ is defined using Equations \ref{K} and \ref{Lx}. For instance, routines to minimize or to  find the zeros of $k[(\Mb{K}, \Mb{L})^{'}]=q[(\Mb{K}, \Mb{L})^{'}]-(\Mb{K}, \Mb{L})^{'}$ can be used.

\newpage
\bibliographystyle{apalike}
\bibliography{Biblio_EndogenousNetwork}
\end{document}